\documentclass[a4paper,UKenglish]{lipics-v2016}
\usepackage{microtype}%if unwanted, comment out or use option "draft"

\def\mQ{{\rm \mathbb{Q}}}
\def\mP{\mathbb{P}}
\def\Q{{\rm \mathbb{Q}}}
\newcommand{\Po}{\mathbb{P}}

\def\prcl{\mathbf{1}}

\def\csp{{\rm CSP}}
\def\Csp{{\rm CSP}}
\def\CSP{{\rm CSP}}

\def\Sym{{\rm Sym}}
\def\Aut{{\rm Aut}}
\def\End{{\rm End}}
\def\Pol{{\rm Pol}}
\def\pol{{\rm Pol}}

\def\Low{{\rm Low}}

\def\Abv{{\rm Abv}}

\def\cycl{{\rm Cycl}}
\def\sept{{\rm Sep}}

\def\Par{{\rm Par}}

\def\Betw{{\rm Betw}}
\def\Cycl{{\rm Cycl}}
\def\Sept{{\rm Sep}}
\def\lowbot{ {\bot \atop <} }

\newcommand{\fraisse}{Fra\"iss\'{e}}
\newcommand{\botfall}{$\bot$-falling}

\theoremstyle{plain}{

\newtheorem{proposition}[theorem]{Proposition}
\newtheorem{observation}[theorem]{Observation}
\newtheorem{notation}[theorem]{Notation}

}

\title{A complexity dichotomy for poset constraint satisfaction} %\footnote{This work was partially supported by someone.}}
\titlerunning{A complexity dichotomy for poset constraint satisfaction} %optional, in case that the title is too long; the running title should fit into the top page column

\author[1]{Michael Kompatscher\footnote{The first author as been funded through projects P27600 and I836-N26 of the Austrian Science Fund (FWF).}}
\author[2]{Trung Van Pham\footnote{The second author has received funding from the European Research Council under the European Community's Seventh Framework Programme (FP7/2007-2013 Grant Agreement no. 257039) and project P27600 of the Austrian Science Fund (FWF).}}
\affil[1]{Theory and Logic Group, Technische Universit\"{a}t Wien, Austria\\
  \texttt{michael@logic.at}}
\affil[2]{Theory and Logic Group, Technische Universit\"{a}t Wien, Austria\\
  \texttt{pvtrung@logic.at}}
\authorrunning{M. Kompatscher and T. V. Pham} %mandatory. First: Use abbreviated first/middle names. Second (only in severe cases): Use first author plus 'et. al.'

\Copyright{Michael Kompatscher and Trung Van Pham}%mandatory, please use full first names. LIPIcs license is "CC-BY";  http://creativecommons.org/licenses/by/3.0/

% mandatory: Please choose ACM 1998 classifications from http://www.acm.org/about/class/ccs98-html . E.g., cite as "F.1.1 Models of Computation". 
% mandatory: Please provide 1-5 keywords
\EventEditors{John Q. Open and Joan R. Acces}
\EventNoEds{2}
\EventLongTitle{42nd Conference on Very Important Topics (CVIT 2016)}
\EventShortTitle{CVIT 2016}
\EventAcronym{CVIT}
\EventYear{2016}
\EventDate{December 24--27, 2016}
\EventLocation{Little Whinging, United Kingdom}
\EventLogo{}
\SeriesVolume{42}
\ArticleNo{23}
\begin{document}

\maketitle
\begin{abstract}
In this paper we determine the complexity of a broad class of problems that extends the temporal constraint satisfaction problems in \cite{Temp-SAT}. To be more precise we study the problems Poset-SAT($\Phi$), where $\Phi$ is a given set of quantifier-free $\leq$-formulas. An instance of Poset-SAT($\Phi$) consists of finitely many variables $x_1,\ldots,x_n$ and formulas $\phi_i(x_{i_1},\ldots,x_{i_k})$ with $\phi_i \in \Phi$; the question is whether this input is satisfied by any partial order on $x_1,\ldots,x_n$ or not. We show that every such problem is NP-complete or can be solved in polynomial time, depending on $\Phi$.

All Poset-SAT problems can be formalized as constraint satisfaction problems on reducts of the random partial order. We use model-theoretic concepts and techniques from universal algebra to study these reducts. In the course of this analysis we establish a dichotomy that we believe is of independent interest in universal algebra and model theory. \end{abstract}

%\tableofcontents

%\printindex

\section{Introduction}

Reasoning about temporal knowledge is a common task in various areas of computer science, for example Artificial Intelligence, Scheduling, Computational Linguistics and Operations Research. In many application temporal constraints are expressed as collections of relations between time points or time intervals. A typical computational problem is then to determine whether such a collection is satisfiable or not.

A lot of research in this area concerns only linear models of time. In particular there exists a complete classification of all satisfiability problems for linear temporal constraints in \cite{Temp-SAT}. However, it has been observed many times that more complex time models are helpful, for instance in the analysis of concurrent and distributed systems or certain planning domains. A possible generalizations is to model time by partial orders (e.g. in \cite{lamport1986mutual}, \cite{anger1989lamport}). 

Some cases of the arising satisfiability problems have already been studied in \cite{BJ-pointalgebras}. We will give a complete classification in this paper. Speaking more formally, let $\Phi$ be a set of quantifier-free formulas in the language consisting of a binary relation symbol $\leq$. Then Poset-SAT($\Phi$) asks if constraint expressible in $\Phi$ are satisfiable by a partial order or not. We are going to give a full complexity classification of problems of the type Poset-SAT($\Phi$). In particular we are going to show that every such problem is NP-complete or solvable in polynomial time.

The proof of our result is based on a variety of methods and results. A first step is that we give a description of every Poset-SAT problem as \textit{constraint satisfaction problem} over a countably infinite domain, where the constraint relations are first-order definable over the \textit{random partial order}, a well-known structure in model theory.

A helpful result has already been established in the form of the classification of the closed supergroups of the automorphism group of the random partial order in \cite{Poset-Reducts}. We extend this analysis to closed transformation monoids. Informally, then our result implies that we can identify three types of Poset-SAT problems: (1) trivial ones (i.e., if there is a solution, there is a constant solution), (2) problems that can be reduced to the problems studied in \cite{Temp-SAT} and (3) CSPs on templates that are model-complete cores. 

So we only have to study problems in the third class. The basic method to proceed then is the \textit{universal-algebraic approach} to constraint satisfaction problems. Here, one studies certain sets of operations (known as \textit{polymorphism clones}) instead of analysing the constraints themselves. An important tool to deal with polymorphisms over infinite domains is Ramsey theory. We need a Ramsey result for partially ordered sets from \cite{ramsey-poset} for proving that polymorphisms behave regularly on large parts of their domain. This allows us to perform a more simplified combinatorial analysis.

This paper has the following structure: In Section \ref{sect:preliminaries} we introduce some basic notation and show how every Poset-SAT problem is equal to a constraint satisfaction problem on a reduct of the random partial order. In Section \ref{sect:univalg} we give a brief introduction to the universal-algebraic approach and the methods from Ramsey theory that we need for our classification. Section \ref{sect:cores} contains a preclassification, by the analysis of closed transformation monoids containing the automorphism group of the random partial order. This is followed by the actual complexity analysis using the universal algebraic approach. In Section \ref{sect:mainresult} we summarize our results to show the complexity dichotomy for Poset-SAT problems.

We further show in Section \ref{sect:mainresult} that an even stronger dichotomy holds, regarding the question whether certain reducts of the random partial order allow pp-interpretations of all finite structures or not (cf. the discussion in \cite{topological-birkhoff} and \cite{BP-siggers}). In this respect the situation is similar to previous classifications for CSPs where the constraints are first-order definable over the rational order \cite{Temp-SAT}, the random graph \cite{BP-graphsat}, or the homogeneous binary branching C-relation \cite{BPPPhyloCSP}.
\section{Preliminaries} \label{sect:preliminaries}
In this section we fix some standard terminology and notation. When working with relational structures it is often convenient not to distinguish between a relation and its relational symbol. We will also do so on several occasions, but this should never cause any confusion.

Let $\leq$ always denote a partial order relation, i.e. a binary relation that is reflexive, antisymmetric and transitive. Let $<$ be the corresponding strict order defined by $x \leq y \land x \neq y$. Let $x \bot y$ denote the incomparability relation defined by $\neg(x \leq y) \land \neg(y \leq x)$. Sometimes we will write $x < y_1 \cdots y_n$ for the conjunction of the formulas $x < y_i$ for all $1 \leq i \leq n$. Similarly we will use $x \bot y_1\cdots y_n$ if $x \bot y_i$ holds for all $1 \leq i \leq n$.

\subsection{Poset-SAT($\Phi$) }

Let $\phi(x_1,\ldots,x_n)$ be a formula in the language that only consists of the binary relation symbol $\leq$. Then we say that $\phi(x_1,\ldots,x_n)$ is \textit{satisfiable} if there exists a partial order $(A;\leq)$ with elements $a_1,\ldots, a_n$ such that $\phi(a_1,\ldots,a_n)$ holds in $(A;\leq)$. In this case we call $(A;\leq)$ a \textit{solution} to $\phi$.

Let $\Phi = \{\phi_1,\phi_2,\ldots, \phi_k \}$ be a finite set of quantifier free $\leq$-formulas. Then the poset satisfiability problem Poset-SAT($\Phi$) is the following computational problem:\\

\noindent
\fbox{
  \parbox{\textwidth}{
\textbf{Poset-SAT($\Phi$):}\\
\textsc{Instance}: A finite set of variables $\{ x_1,\ldots,x_n \}$ and a finite set of formulas $\Psi$ that is obtained from $\phi \in \Phi$, by substituting the variables of $\phi$ by variables from $\{ x_1,\ldots,x_n \}$\\
\textsc{Question}: Is there a partial order $(A;\leq)$ that is a solution to all formulas in $\Psi$?
  }
}\ \newline

\begin{example} An instance of Poset-SAT($\{<\}$) is given by variables in $\{x_1,\ldots,x_n\}$ and formulas in $\Psi$ of the form $x_i < x_j$. The question is, if there is partial order on $\{x_1,\ldots,x_n\}$ that satisfies all formulas $x_i < x_j$ in $\Psi$. It is easy to see that such a partial order always exists if $\Psi$ does not contain formulas $x_{i_1} < x_{i_2}, \ldots$, $x_{i_{n-1}} < x_{i_n}$, $x_{i_n} < x_{i_1}$. The existence of such ``cycles'' in $\Psi$ can be verified in polynomial time, thus Poset-SAT($\{<\}$) is tractable.
\end{example}

Every partial order can be extended to a total order. Therefore there is a solution to $I$ if and only if there is a totally ordered solution to $I$. So Poset-SAT($\{<\}$) is the same problem as the corresponding \textit{temporal constraint satisfaction problem} Temp-SAT($\{<\}$), i.e. the question if there is a total order that is a solution to the input.

\begin{example}We define the betweenness relation $\Betw(x,y,z): = (x < y \land y < z) \lor (z < y \land y < x)$. Again an instance of Poset-SAT($\{\Betw\}$) is accepted if and only if it is accepted by Temp-SAT($\{\Betw \}$). It is know that Temp-SAT($\{\Betw\}$) is NP-complete, so Poset-SAT($\{\Betw\}$) is NP-complete. We remark that if every relation in $\Phi$ is positively definable in $<$ that then Poset-SAT($\Phi$) is the same problem as Temp-SAT($\Phi$). All temporal constraint satisfaction problems are classified in \cite{Temp-SAT}.
\end{example}

\begin{example}Let every formula $\phi(x_1,\ldots,x_n)$ in $\Phi$ be a $\leq$-Horn formula, that is a formula of the form
\begin{align*}
x_{i_1} \leq x_{j_1} \land x_{i_2} \leq x_{j_2} \land \cdots \land x_{i_k} \leq x_{j_k} &\to x_{i_{k+1}} \leq x_{j_{k+1}} \text{ or} \\
x_{i_1} \leq x_{j_1} \land x_{i_2} \leq x_{j_2} \land \cdots \land x_{i_k} \leq x_{j_k} &\to \text{ 'false'}
\end{align*}
All problems of this ``Horn-type'' are tractable. One can see this by giving an algorithm based on the resolution rule. We discuss this class of tractable problems in Section \ref{sect:horn}.
\end{example}

\begin{example}We define $\Cycl(x,y,z)$ by
\begin{align*}
\Cycl(x,y,z) := & (x < y \land y < z) \lor (z < x \land x < y) \lor (y < z \land z < x) \lor \\ 
&(x < y \land x \bot z \land y \bot z) \lor (y < z \land x \bot y \land x \bot z) \lor (z < x \land y \bot z \land y \bot x).
\end{align*}
We will see in Section \ref{sect:cycl} that Poset-SAT($\{\Cycl\}$) is an NP-complete problem.
\end{example}

Every problem Poset-SAT($\Phi$) is clearly in NP. We can ``guess'' a partial order on the variables $x_1,\ldots,x_n$ and then checks in polynomial time if this partial order is a solution to the input formulas in $\Psi$ or not. The main result of this paper is to give a full classification of the computational complexity, showing the following dichotomy:

\begin{theorem} \label{theorem:main}
Let $\Phi$ be a finite set of quantifier-free $\leq$-formulas. The problem Poset-SAT($\Phi$) is in P or NP-complete.
\end{theorem}

\subsection{Poset-SAT($\Phi$) as CSP}

In this section we are going to show that every Poset-SAT problem can be translated into a \textit{constraint satisfaction problem} or \textit{CSP} over an infinite domain. This reformulation will allow us the use of universal-algebraic and Ramsey-theoretical methods.

Let $\Gamma$ be a relational structure with signature $\tau = \{R_1,R_2,\ldots\}$, i.e. $ \Gamma = (D;R_1^{\Gamma},R_2^{\Gamma},\ldots)$ where $D$ is the domain of $\Gamma$ and $R_i^{\Gamma} \subset D^{k_i}$ is a relation of arity $k_i$ over $D$. When $\Delta$ and $\Gamma$ are two $\tau$-structures, then a \textit{homomorphism} from $\Delta$ to $\Gamma$ is a mapping $h$ from the domain of $\Delta$ to the domain of $\Gamma$ such that for all $R \in \tau$ and for all $(x_1,\ldots,x_j) \in R^{\Delta}$ we have $h(x_1,\ldots,x_j) \in R^{\Gamma}$. Injective homomorphisms that also preserve the complement of each relation are called \textit{embeddings}. Bijective embeddings from $\Delta$ onto itself are called \textit{automorphisms} of $\Delta$.

Suppose that the signature $\tau$ of $\Gamma$ is finite. Then the constraint satisfaction problem $\csp(\Gamma)$ is the following decision problem:\\

\noindent
\fbox{
  \parbox{\textwidth}{
\textbf{\csp($\Gamma$):}\\
\textsc{Instance}: A finite $\tau$-structure $\Delta$\\
\textsc{Question}: Is there a homomorphism from $\Delta$ to $\Gamma$?
  }
}\ \newline

We say that $\Gamma$ is the \textit{template} of the constraint satisfaction problem $\csp(\Gamma)$.\\ 

We now need some terminology from model theory. Let $\Delta$ be a structure with domain $D$. We say a relation $R \subset D^n$ is definable in $\Delta$, if there is a first order formula $\phi$ in the signature of $\Delta$ with free variables $(x_1,\ldots,x_n)$ such that $(a_1,\ldots,a_n) \in R$ if and only if $\phi(a_1,\ldots,a_n)$ holds in $\Delta$. A relational structure $\Gamma$ with the same domain as $\Delta$ is called a \textit{reduct} of $\Delta$ if all relations of $\Gamma$ are definable in $\Delta$.

A structure is called \textit{homogeneous} if every isomorphism between finitely generated substructures can be extended to an automorphism of the whole structure. 

Let $\mathcal C$ be a class of relational structures of the same signature. We say $\mathcal C$ has the \textit{amalgamation property} if for every $A,B,C \in \mathcal C$ and all embeddings $u:A \to B$ and $v: A \to C$ there is a $D \in \mathcal C$ with embeddings $u': B \to D$, $v': C \to D$ such that $u' \circ u = v' \circ v$. A class $\mathcal C$ is called an \textit{amalgamation class} if it has the amalgamation property and is closed under isomorphism and taking induced substructures.

\begin{theorem}[\fraisse, see Theorem 7.1.2 in \cite{Hodges}] \label{theorem:fraisse}
Let $\mathcal C$ be an amalgamation class that has only countably many non-isomorphic members. Then there is a countable homogeneous structure $\Gamma$ such that $\mathcal C$ is the \textit{age} of $\Gamma$, i.e. the class of all structures that embeds into $\Gamma$. The structure $\Gamma$, which is unique up to isomorphism, is called the \fraisse{} limit of $\Gamma$. \hfill $\square$
\end{theorem}

Since the class of all finite partial orders forms an amalgamation class, it has a \fraisse{} limit which is called the \textit{random partial order} or \textit{random poset} $\Po = (P;\leq)$. The random poset is a well-studied object in model theory that has a lot of helpful properties. As a homogenous structure in a finite relational language $\Po$ has \textit{quantifier-elimination}, i.e. every formula in $\Po$ is equivalent to a quantifier-free formula. Also it is \textit{$\omega$-categorical}, i.e. all countable structures that satisfy the same first order formulas as $\Po$ are isomorphic to $\Po$. Reducts of $\omega$-categorical structures are $\omega$-categorical as well. Therefore $\Po$ and all of its reducts are $\omega$-categorical. For further model-theoretic background on $\omega$-categorical and homogeneous structures we refer to \cite{Hodges}. \\

%We will work from now on with the random poset. Reducts of $\omega$-categorical structures are $\omega$-categorical as well. Therefore $\Po$ and all of its reducts are $\omega$-categorical.  For further model-theoretic background on $\omega$-categorical and homogeneous structures we refer to \cite{Hodges}.

Now let $\Phi = \{\phi_1,\ldots, \phi_n\}$ be a finite set of quantifier free $\leq$-formulas. We associate with $\Phi$ the $\tau$-structure $\Po_\Phi = (P;R_1,\ldots,R_n)$ that we obtain by setting $(a_1,\ldots,a_k) \in R_i$ if and only if $\phi_i(a_1,\ldots,a_k)$ holds in $\Po$. We claim that $\csp(\Gamma_\Phi)$ and Poset-SAT($\Phi$) are essentially the same problem. Let $\Psi$ be an instance of Poset-SAT($\Phi$) with the variables $x_1,\ldots,x_n$. Then we define a structure $\Delta_{\Psi}$ with domain $\{x_1,\ldots,x_n\}$ as follows: The relation $R_i^{\Delta}$ contains the tuple $(x_{i_1},\ldots,x_{i_{k_i}})$ if and only if the instance $\Psi$ contains the formula $\phi_i(x_{i_1},\ldots,x_{i_{k_i}})$. It is quite straightforward to see that $\Psi$ is accepted by Poset-SAT($\Phi$) if and only if $\Delta_{\Psi}$ homomorphically maps to $\Po_\Phi$.

Conversely let $\Gamma = (P;R_1,\ldots,R_n)$ be a reduct of $\Po$. Since $\Po$ has quantifier-elimination, every relation $R_i$ in $\Gamma$ has a quantifier-free definition $\phi_i$ in $\Po$. With that in mind, we study the problem Poset-SAT($\{\phi_1,\ldots,\phi_n\}$). Let $\Delta$ be an instance of $\csp(\Gamma)$. Then let $\Psi_\Delta$ be an instance of Poset-SAT($\Phi$) where the variables are the points in $\Delta$ and $\phi_i(x_1,\ldots,x_k) \in \Psi$ if and only if $(x_1,\ldots,x_k) \in R_i^{\Delta}$. Then $\Delta$ is accepted by $\csp(\Gamma)$ if and only if $\Psi_\Delta$ is accepted by Poset-SAT($\Phi$).

So by the observations in the paragraphs above the following holds:

\begin{proposition}
The problems of the form \textnormal{Poset-SAT($\Phi$)} correspond precisely to the problems of the form $\CSP(\Gamma)$, where $\Gamma$ is a reduct of the random partial order $\Po$. \hfill $\square$
\end{proposition}

%It is easy to find a structure $\Gamma_{\Phi}$ such that an instance $\Psi$ of Poset-SAT($\Phi$) is accepted if and only if $\Delta_{\Psi}$ homomorphically maps to $\Gamma_{\Phi}$. In fact $\Gamma_{\Phi}$ can be obtained by taking any infinite structure $\Gamma = (D; \leq)$ whose \textit{age}, i.e. the class of finite structures that embed into $\Gamma$, consists all finite partial orders. 

\section{The universal-algebraic approach} \label{sect:univalg}

We apply the so-called universal-algebraic approach and the Ramsey theoretical methods developed by Bodirsky and Pinsker in \cite{BP-graphsat} to obtain our complexity results. Using the language of universal algebra we can elegantly describe the border between tractability and NP-hardness for CSPs on reducts of the random poset. In this section we give a brief introduction to this approach. For a more detailed introduction we refer to \cite{manuelhabil} as well as the shorter \cite{pinsker-surveyCSP}.

\subsection{Primitive positive definability}

A  first-order formula $\phi(x_1,\ldots, x_n)$ in the language $\tau$ is called \textit{primitive positive} if it is of the form $\exists y_1,\ldots, y_k~(\psi_1 \land \cdots \land \psi_m)$ where $\psi_1, \ldots, \psi_m$ are atomic $\tau$-formulas with free variables from the set $\{x_1,\ldots, x_n, y_1,\ldots, y_k\}$.

Let $\Gamma$ be a $\tau$-structure. We then say a relation $R$ is is \textit{primitively positive definable} or \textit{pp-definable} in $\Gamma$ if there is a primitive positive formula $\phi(x_1,\ldots, x_n)$ such that $(a_1,\ldots,a_n) \in R$ if and only if $\phi(a_1,\ldots, a_n)$ holds in $\Gamma$.

\begin{lemma}[Jeavons \cite{jeavons1998algebraic}] \label{lemma:pp}
Let $\Gamma$ be a relational structure in finite language, and let $\Gamma'$ be the structure obtained from $\Gamma$ by adding a relation $R$. If R is primitive positive definable in $\Gamma$, then $\csp(\Gamma)$ and $\csp(\Gamma')$ are polynomial-time equivalent. \hfill $\square$
\end{lemma}

By $\langle \Gamma \rangle_{pp}$ we denote the set of all primitively positive definable relations on $\Gamma$.
So for two structures $\Gamma$ and $\Delta$ the problems $\csp(\Gamma)$ and $\csp(\Delta)$ have the same complexity if $\langle \Gamma \rangle_{pp} = \langle \Delta \rangle_{pp}$. This means that in our analysis we only have to study reduct of the random poset up to primitive positive definability.

\subsection{Polymorphism clones}

Let $\Gamma$ be a relational structure with domain $D$. By $\Gamma^n$ we denote the direct product of $n$-copies of $\Gamma$. This is, we take a structure on $D^n$ with same signature $\Gamma$. Then for $n$-tuples $\bar x^{(1)},\ldots, \bar x^{(k)}$ we set that $(\bar x^{(1)},\ldots, \bar x^{(k)}) \in R$ if and only if $(x^{(1)}_i,\ldots x^{(k)}_i) \in R$ holds in $\Gamma$ for every coordinate $i \in [n]$.

Then an $n$-ary operation $f$ is called a \textit{polymorphism} of $\Gamma$ if $f$ is a homomorphism from $\Gamma^n$ to $\Gamma$. Unary polymorphisms are called \textit{endomorphisms}. For every relation $R$ on $D$ we say $f$ \textit{preserves} $R$ if $f$ is a polymorphism of $(D;R)$. Otherwise we say $f$ \textit{violates} $R$.

%The set of all polymorphisms of a structure $\Gamma$ is called the \textit{polymorphism clone} of $\Gamma$, or short $\Pol(\Gamma)$. The set of all endomorphisms of a structure $\Gamma$ is called the \textit{endomorphism monoid} $\End(\Gamma)$

For a given structure $\Gamma$ the set of all polymorphisms $\Pol(\Gamma)$ contains all the projections $\pi_i^n(x_1,\ldots,x_n) = x_i$ and is closed under composition. Every set of operation with these properties is called a \textit{clone} or \textit{function clone} (cf. \cite{szendreiclones}). $\Pol(\Gamma)$ is called the \textit{polymorphism clone} of $\Gamma$. We write $\Pol(\Gamma)^{(k)}$ for the set of $k$-ary functions in $\Pol(\Gamma)$. We write $\End(\Gamma)$ for the monoid consisting of all endomorphisms of $\Gamma$.\\
The clone $\Pol(\Gamma)$ is furthermore \textit{locally closed} in the following sense: Let $k > 1$ be arbitrary and $f$ a $k$-ary operation on $D$. If for every finite subset $A \subseteq D^k$ there is a $g \in \Pol(\Gamma)$ with $g \restriction A = f \restriction A$ then also $f \in \Pol(\Gamma)$. 

For a set of operation $F$ on $D$ being locally closed is equal to be closed in the topology of pointwise convergence. We write $\overline F$ for smallest locally closed set containing $F$. We say a set of operation $F$ \textit{generates} an operation $g$ if $g $ is in the smallest locally closed clone containing $F$. 

It is of central importance to us that primitive positive definability in $\omega$-categorical (and finite) structures can be characterized by preservation under polymorphisms:

\begin{theorem}[Bodirsky, Ne{\v{s}}et{\v{r}}il \cite{BN-ppdefinability}]
Let $\Gamma$ be an $\omega$-categorical structure. Then a relation is pp-definable in $\Gamma$, if and only if it is preserved by the polymorphisms of $\Gamma$.
\end{theorem}

Thus, by Lemma \ref{lemma:pp} the complexity of $\csp(\Gamma)$ only depend on the polymorphism clone $\Pol(\Gamma)$ for $\omega$-categorical $\Gamma$.

We also need the fact that a relation is not definable in $\Gamma$ can be described by polymorphisms of bounded arity.
\begin{theorem}[Bodirsky, Kara \cite{Temp-SAT}] \label{theorem:violate}
Let $\Gamma$ be a relational structure and let $R$ be a $k$-ary relation that is a union of at most $m$ orbits of $\Aut(\Gamma)$ on $D^k$. If $\Gamma$ has a polymorphism $f$ that violates $R$, then $\Gamma$ also has an at most $m$-ary polymorphism that violates $R$. \hfill $\square$
\end{theorem}

\subsection{Structural Ramsey theory}

We apply Ramsey theory to analyse certain regular behaviour of polymorphisms. This approach was invented by Bodirsky and Pinsker and has been proven to be useful in several other classification results (e.g. \cite{Temp-SAT}, \cite{BP-graphsat}, \cite{BPPPhyloCSP}). We only give a brief introduction how Ramsey theory helps us to study polymorphism clones, a detailed introduction can be found in \cite{BP-reductsRamsey}.

Let $\Gamma$ and $\Delta$ be two finite structures of the same signature. Then $\binom{\Delta}{\Gamma}$ denotes all the substructures of $\Delta$ that are isomorphic to $\Gamma$. We write $\Theta \to (\Delta)^\Gamma_k$ if for every coloring $\kappa: \binom{\Theta}{\Gamma} \to [k]$ there is an isomorphic copy $\Delta'$ of $\Delta$ such that $\kappa$ is monochromatic on $\binom{\Delta'}{\Gamma}$.

Let $\mathcal C$ be a category of structures with the same signature that is closed under taking isomorphic copies and substructures.  Then $\mathcal C$ is said to be a \textit{Ramsey class} if for every $k \geq 1$ and every $\Gamma, \Delta \in \mathcal C$ there is a $\Theta \in \mathcal C$ such that $\Theta \to (\Delta)^\Gamma_k$. 

A homogeneous structure is said to be a \textit{homogeneous Ramsey structure} if its age is a Ramsey class. A structure is called \textit{ordered}, if it contains a total order relation.

If we look at the class of all structures $(A;\leq, \prec)$, where $\prec$ is a linear order that extends $\leq$, it has the Ramsey property and the amalgamation property (see \cite{ramsey-poset}). Therefore its \fraisse{} limit $(P;\leq,\prec)$ is an ordered homogeneous Ramsey structure.

\begin{definition}
Let $\Gamma$ and $\Delta$ be two structures and $f$ an $n$-ary operation from the domain of $\Gamma$ to the domain of $\Delta$. Let $A$ be a subset of $\Gamma$. Then $f$ is called \textnormal{canonical on $A$}, if for every integer $s \geq 1$, all automorphisms $\alpha_1, \ldots, \alpha_n \in \Aut(\Gamma)$ and tuples $d_1,\ldots,d_n \in A^s$ there is an automorphism $\beta \in \Aut(\Delta)$ such that
\[ f(\alpha_1(d_1), \ldots, \alpha_n(d_n)) = \beta(f(d_1,\ldots,d_n)). \]
If $f$ is canonical on the domain of $\Gamma$, we say that $f$ is \textnormal{canonical}.
\end{definition}

We remark that sometimes canonical functions are also defined as functions that preserve model-theoretic types. However, by the theorem of Ryll-Nardzewski, Engeler and Svenonius (confer \cite{Hodges}), the $s$-types of a countable $\omega$-categorical structure $\Gamma$ are exactly the orbits of $s$-tuples under the action of $\Aut(\Gamma)$. Therefore the model theoretical definition is equivalent to the one we gave above, for $\omega$-categorical $\Gamma$ and $\Delta$. For this reason we are also going to use (s-)orbit and (s-)type synonymously.

Now Ramsey structures allows us to generate functions that are canonical on arbitrary large finite substructures of $\Gamma$ in the following sense:

\begin{theorem}[Lemma 19 in \cite{BP-reductsRamsey}] \label{theorem:ramsey2}
Let $\Gamma$ be an ordered homogeneous Ramsey structure with domain $D$ and let $f: D^l \to D$. Then for every finite subset $A$ of $\Gamma$ there are automorphisms $\alpha_1,\ldots,\alpha_l \in \Aut(\Gamma)$ such that $f \circ (\alpha_1,\ldots,\alpha_l)$ is canonical on $A^l$. \hfill $\square$
\end{theorem} 

We can refine this statement by additionally fixing constants. Let $c_1,\ldots,c_n$ be elements of the domain of $\Gamma$. Then $(\Gamma,c_1,\ldots,c_n)$ denotes the structure that we obtain by extending $\Gamma$ by the constants $c_1,\ldots,c_n$.

\begin{theorem}[Lemma 27 in \cite{BP-reductsRamsey}] \label{theorem:ramsey}
Let $\Gamma$ be an ordered homogeneous Ramsey structure with domain $D$. Let $c_1,\ldots,c_n \in D$ and $f: D^l \to D$. Then $\{f\} \cup \Aut(\Gamma,c_1,\ldots,c_n)$ generates a function $g$ that is canonical as operation from $( \Gamma, c_1, \ldots, c_n)$ to $\Gamma$ and satisfies $f \restriction \{c_1,\ldots,c_n\} = g \restriction \{c_1,\ldots,c_n\}$. \hfill $\square$
\end{theorem}

By the \textit{behaviour} of a canonical function $f: \Delta \to \Lambda$ we denote the set of all tuples $(p,q)$ where $p$ is an $s$-type of $\Delta$, $q$ is a $s$-type of $\Lambda$ and for for every tuple $\bar a$ of type $p$ the image $f(\bar a)$ has type $q$ in $\Lambda$. So we can regard the behaviour of a canonical function as a function from the types of $\Delta$ to the types of $\Lambda$. Let $A$ be a subset of the domain of $\Delta$. We call the behaviour of $f \restriction A$ the \textit {behaviour of $f$ on $A$}. We say a function $f: \Delta \to \Lambda$ \textit{behaves like} $g: \Delta \to \Lambda$ (on $A$) if their behaviour (on $A$) is equal.

\subsection{Model-complete cores}
Let $\Delta$ and $\Gamma$ be to structures with the same signature. We say $\Delta$ is \textit{homomorphically equivalent} to $\Gamma$ if there is a homomorphisms from $\Delta$ to $\Gamma$ and a homomorphism from $\Gamma$ to $\Delta$. By definition, the constraint satisfaction problems $\csp(\Delta)$ and $\csp(\Gamma)$  encode the same computational problem for homomorphically equivalent structures $\Delta$ and $\Gamma$. By homomorphic equivalence we can find for every CSP a template with some nice model-theoretical properties.

A structure $\Delta$ is called a \textit{core} if every endomorphism $e$ of $\Delta$ is a self embedding. A structure is called \textit{model-complete} if every formula in its first order theory is equivalent to an existential formula. Note that endomorphisms preserve existential positive formulas, and embeddings preserve existential formulas. In the case of $\omega$-categorical structures also the opposite holds. Therefore we have:

\begin{lemma}[Lemma 13 in \cite{BP-minimal}]
A countable $\omega$-categorical structure $\Gamma$ is a model-complete core if and only if the endomorphism monoid of $\Gamma$ is generated by $\Aut(\Gamma)$. In this case, every definable relation in $\Gamma$ is also definable by existential positive formulas. \hfill $\square$
\end{lemma}

Now every CSP with $\omega$-categorical template can be reformulated as a CSP on a template with model-complete core by the following theorem:

\begin{theorem}[Theorem 16 from \cite{manuel-core}]
Every $\omega$-categorical structure $\Gamma$ is homomorphically equivalent to a model-complete core which is unique up to isomorphism. This core is $\omega$-categorical or finite. \hfill $\square$
\end{theorem}

So an important step in analyzing the complexity of $\csp(\Gamma)$ is to identify the model-complete core of $\Gamma$. By the following theorem of Bodirsky the complexity of the CSP of a core does not increase if we add finitely many constants.

\begin{theorem}[Theorem 19 from \cite{manuel-core}] \label{theorem:coreconst}
Let $\Gamma$ be a model-complete $\omega$-categorical or finite core, and let $c$ be an element of $\Gamma$. Then $\csp(\Gamma)$ and $\csp(\Gamma,c)$ have the same complexity, up to polynomial time. \hfill $\square$
\end{theorem}

\subsection{Primitive positive interpretations} \label{sect:topclones}

A tool to compare the complexity of CSPs of structures $\Delta$ and $\Gamma$ of possibly different domains and signatures are interpretations. We say $\Delta$ is \textit{pp-interpretable} in $\Gamma$ if there is a $n \geq 1$ and a partial map $I:\Gamma^n \to \Delta$ such that 
\begin{itemize}
\item $I$ is surjective,
\item the domain of $I$ is pp-definable in $\Gamma$,
\item the preimage of the equality relation in $\Delta$ is pp-definable in $\Gamma$, 
\item the preimage of every relation in $\Delta$ is pp-definable in $\Gamma$.
\end{itemize}

Then the following result holds:
\begin{lemma}[Theorem 5.5.6 in \cite{manuelhabil}]
If $\Delta$ is pp-interpretable in $\Gamma$ then $\csp(\Delta)$ can be reduced to $\csp(\Gamma)$ in polynomial time. 
\end{lemma}

The positive-1-in-3-SAT problem is a well-studied problem in literature that is known to be NP-complete \cite{schaefer}. It can be written as $\csp(\{0,1\};\rm{1IN3})$ with $\rm{1IN3}:= \{(1,0,0),(0,1,0),(0,0,1)\}$. In practice we often show the NP-completeness of $\csp(\Gamma)$ by finding a pp-interpretation of $(\{0,1\};\rm{1IN3})$ in $\Gamma$.

As pp-definability, also pp-interpretation can be translated into the language of clones. A \textit{clone homomorphism} $\xi:\Pol(\Gamma) \to \Pol(\Delta)$ is a map that preserves arities such that
\begin{itemize}
\item $\xi(\pi_i^n) = \pi_i^n$ for all projections,
\item $\xi(f \circ (g_1,\ldots,g_n)) = \xi(f) \circ (\xi(g_1),\ldots, \xi(g_n))$.
\end{itemize}
For $\omega$-categorical or finite structures $\Delta$ and $\Gamma$ it is known that $\Delta$ is pp-interpretable in $\Gamma$  if and only if there exists a continuous \textit{clone homomorphism} from $\Pol(\Gamma)$ onto $\Pol(\Delta)$.

\begin{theorem}[Bodirsky, Pinsker \cite{topological-birkhoff}] \label{theorem:topbirk}
Let $\Gamma$ be finite or $\omega$-categorical and $\Delta$ be finite. Then $\Gamma$ has a primitive positive interpretation in $\Delta$ if and only if $\Pol(\Gamma)$ has a continuous clone homomorphism to $\Pol(\Delta)$. \hfill $\square$
\end{theorem}

The polymorphism clone of $(\{0,1\};\rm{1IN3})$ consist only of the projections $\pi_i^n$ on the two element set $\{0,1\}$. We call it the \textit{projection clone} $\prcl$. So $(\{0,1\};\rm{1IN3})$ has a pp-interpretation in the $\omega$-categorical structure $\Gamma$ if and only if there is a continuous clone homomorphism $\xi: \Pol(\Gamma) \to \prcl$. By Theorem \ref{theorem:topbirk} then \emph{every} finite structure $\Delta$ has a pp-interpretation in $\Gamma$, since it is easy to see that there is always a continuous clone homomorphism from $\prcl$ to $\Pol(\Delta)$.\\

We summarize that for finite or $\omega$-categorical and finite structures $\Delta$ and $\Gamma$ the complexity of $\csp(\Delta)$ reduces to $\csp(\Gamma)$ in the following cases:
\begin{enumerate}
\item $\Delta$ is pp-interpretable in $\Gamma$.
\item $\Delta$ is the model-complete core of $\Gamma$.
\item $\Gamma$ is a model-complete core and $\Delta$ is obtained by adding finitely many constants to the signature of $\Gamma$.
\end{enumerate}

%So if we can obtain $(\{0,1\};\rm{1IN3})$ from $\Gamma$ by the reductions (1)-(3) we know that $\Gamma$ is NP-complete. 

By some recent results of Barto, Pinsker we can also express this fact nicely by polymorphism clones. %We say that a map $\xi: \Pol(\Gamma) \to \Pol(\Delta)$ is a \textit{h1 clone homomorphism} if it preserves arities and 
%\[\xi(f(\pi_{i_1},\ldots,\pi_{i_n})) = \xi(f)(\pi_{i_1},\ldots,\pi_{i_n}) \] 
%holds for all $n \geq 1$ and $f \in \Pol(\Gamma)$ and every indices $i_1,\ldots,i_n \in [n]$. Then the following holds:

\begin{theorem}[Theorem 1.4 in \cite{BP-siggers}] \label{theorem:doubleshrink}
Let $\Gamma$ be finite or $\omega$-categorical and let $\Delta$ be its model-complete core. Then the following are equivalent:
\begin{enumerate}
\item Every finite structure has a pp-interpretation in some extension of $\Delta$ by finitely many constants.
%\item There is a uniformly continuous h1 clone homomorphism $\Pol(\Gamma) \to \prcl$.
\item There is a continuous clone homomorphism $\Pol(\Delta,c_1,\ldots,c_n) \to \prcl$ for some constants $c_1,\ldots,c_n \in \Delta$.
\item There is a clone homomorphism $\Pol(\Delta,c_1,\ldots,c_n) \to \prcl$ for some $c_1,\ldots,c_n \in \Delta$.
\item $\Pol(\Delta)$ contains no \textnormal{pseudo-Siggers operation}, i.e. a 6-ary operations $s$ such that
\[ \alpha s(x,y,x,z,y,z) = \beta s(y,x,z,x,z,y) \]
for some unary $\alpha,\beta \in \Pol(\Delta)$. \hfill $\square$
\end{enumerate}
\end{theorem}

Note that (4) shows that the question whether $(\{0,1\};\rm 1IN3)$ is pp-interpretable in any $(\Delta,c_1,\ldots,c_n)$ only depends on algebraic but not the topological properties of $\Pol(\Delta)$.

\section{A pre-classification by model-complete cores} \label{sect:cores}

In this section we start our analysis of reducts of the random partial order $\Po = (P;\leq)$. Our aim is to determine the model-complete core for every reduct $\Gamma$ of $\Po$. Therefore we are going to study all possible endomorphism monoids $\End(\Gamma) \supseteq \Aut(\Po)$. Part of the work was already done in \cite{Poset-Reducts} where all the automorphism groups $\Aut(\Gamma) \supseteq \Aut(\Po)$ were determined. Several parts of our proof are very similar to the group case; at that points we are going to directly refer to the corresponding proofs of \cite{Poset-Reducts}.\\

If we turn the partial order $\Po$ upside-down, then the obtained partial order is again isomorphic to $\Po$. Hence there exists a bijection $\updownarrow: P \to P$ such that for all $x,y \in P$ we have $x < y$ if and only if $\updownarrow (y) < \updownarrow (x)$. By the homogeneity of $\Po$ it is easy to see that the monoid generated by $\updownarrow$ and $\Aut(\Po)$ does not depend on the choice of the function $\updownarrow$.

The class of all finite structures $(A;\leq,F)$, where $(A;\leq)$ is a partial order and $F$ is upwards closed set is an amalgamation class.  Its \fraisse{} limit is isomorphic to $\Po$ with an additional unary relation $F$. We say $F$ is a \textit{random filter} on $\Po$. Note that $F$ and $I = P \setminus F$ are both isomorphic to the random partial order. Furthermore for every pair $x \in I$ and $y \in F$ either $x < y$ or $x \bot y$ holds. 

We define a new order relation $<_F$ on by setting $x <_F y$ if and only if
\begin{itemize}
\item $x,y \in F$ and $x < y$ or,
\item $x,y \in I$ and $x < y$ or,
\item $x \in I$, $y \in F$ and $x \bot y$.
\end{itemize}

It is shown in \cite{Poset-Reducts} that the resulting structure $(P;<_F)$ is isomorphic to $(P,<)$. We fix a map $\circlearrowright_F: P \to P$ that maps $(P;<)$ isomorphically to $(P,<_F)$. By the homogeneity of $\Po$ one can see that the smallest closed monoid generated by $\circlearrowright$ and $\Aut(\Po)$ does not depend on the choice of the random filter $F$. We fix a random filter $F$ and set $\circlearrowright := \circlearrowright_F$.

For $B \subseteq \Sym(P)$, let $\langle B \rangle$ denote the smallest closed subgroup of $\Sym(P)$ containing $B$. For brevity, when it is clear we are discussing supergroups of $\Aut(\Po)$, we may abuse notation and write $\langle B \rangle$ to mean $\langle B \cup \Aut(\Po) \rangle$.

\begin{theorem}[Theorem 1 from \cite{Poset-Reducts}] \label{theorem:groups}
Let $\Gamma$ be a reduct of $\Po$. Then $\Aut(\Gamma) \supseteq \Aut(\Po)$ is equal to one of the five groups
$\Aut(\Po)$, $\langle \updownarrow\rangle$, $\langle \circlearrowright\rangle$, $\langle \updownarrow, \circlearrowright\rangle$ or $\Sym(P)$. \hfill $\square$ 
\end{theorem}

We are going to show the following extension of Theorem \ref{theorem:groups}:
\begin{proposition} \label{theorem:monoids}
Let $\Gamma$ be a reduct of $\Po$. Then for $\End(\Gamma)$ at least one of the following cases applies:
\begin{enumerate}
\item $\End(\Gamma)$ contains a constant function,
\item $\End(\Gamma)$ contains a function $g_<$ that preserves $<$ and maps $P$ onto a chain,
\item $\End(\Gamma)$ contains a function $g_{\bot}$ that preserves $\bot$ and maps $P$ onto an antichain,
\item The automorphism group $\Aut(\Gamma)$ is dense in $\End(\Gamma)$, i.e. $\Gamma$ is a model-complete core. So by the classification in Theorem \ref{theorem:groups}, $\End(\Gamma)$ is the topological closure of $\Aut(\Po)$, $\langle \updownarrow \rangle$, $\langle \circlearrowright \rangle$, $\langle \updownarrow, \circlearrowright \rangle$ or $\Sym(P)$ in the space of all functions $P^P$.
\end{enumerate}
\end{proposition}

Before we start with the proof Proposition \ref{theorem:monoids} we want to point out its relevance for the complexity analysis of the CSPs on reducts of $\Po$.

Constraint satisfaction problems on reducts of $(\mQ;<)$ are called \textit{temporal satisfaction problems}. The CSPs on reducts of a countable set with a predicate for equality $(\omega;=)$ are called \textit{equality satisfaction problems}. For both classes a full complexity dichotomy is known, see \cite{Temp-SAT} and\cite{BKeqSAT}. As a corollary of Proposition \ref{theorem:monoids} we get the following pre-classification of CSPs reducing all the cases where $\Gamma$ is not a model-complete core to temporal or equality satisfaction problems:

\begin{corollary} \label{corollary:monoids}
Let $\Gamma$ be a reduct of $\Po$. Then one of the following holds
\begin{enumerate}
\item $\CSP(\Gamma)$ is trivial;
\item The model-complete core of $\Gamma$ is a reduct of $(\omega;=)$,\\ so $\CSP(\Gamma)$ is equal to an equality satisfaction problem;
\item The model-complete core of $\Gamma$ is a reduct of $(\Q;<)$,\\ so $\CSP(\Gamma)$ is equal to a temporal satisfaction problem;
\item $\End(\Gamma)$ is the topological closure of $\Aut(\Po)$, $\langle \updownarrow \rangle$, $\langle \circlearrowright \rangle$ or $\langle \updownarrow, \circlearrowright \rangle$.
\end{enumerate}
\end{corollary}

\begin{proof}
If there is a constant function in $\End(\Gamma)$, then $\CSP(\Gamma)$ accepts every instance, so we are in the first case. So let $\End(\Gamma)$ contain no constants.

Assume that $g_\bot \in \End(\Gamma)$. Since $g_\bot$ preserves $\bot$, the image of $(P;\bot)$ under $g_\bot$ is isomorphic to a countable antichain, or in other word, a countable set $\omega$ with a predicate for inequality $(\omega;\neq)$. Thus, for every reduct of $\Gamma$ the image $g_\bot(\Gamma)$ can be seen as a reduct of $(\omega;\neq)$. Now clearly $\Gamma$ and $g_\bot(\Gamma)$ are homomorphically equivalent. It is shown in \cite{BKeqSAT} that every reduct of $(\omega;\neq)$ without constant endomorphisms is a model-complete core. So we are in the second case.

Now assume that $g_< \in \End(\Gamma)$ but $g_\bot \not\in \End(\Gamma)$. Since $g_<$ preserves $<$ and is a chain, the image of $(P;<)$ under $g_<$ has to be isomorphic to the rational order $(\mQ;<)$. Thus for every reduct of $\Gamma$ the image $g_<(\Gamma)$ can be seen as a reduct of $\mQ$. Now clearly $\Gamma$ and $g_<(\Gamma)$ are homomorphically equivalent. It is shown in \cite{Temp-SAT} that the model-complete core of every reduct of $(\Q,<)$ is either trivial, definable in $(\omega,\neq)$ or the reduct itself. So we are in the third case.

Note that also in the case where $\End(\Gamma) = \overline{\Sym(P)}$ we have that $e_\bot \in \End(\Gamma)$. So by Proposition \ref{theorem:monoids} we are only left with the cases where $\End(\Gamma)$ is the topological closure of $\Aut(\Po)$, $\langle \updownarrow \rangle$, $\langle \circlearrowright \rangle$ or $\langle \updownarrow, \circlearrowright \rangle$.
\end{proof}

Let us define the following relations on $P$:
\begin{align*}
\Betw(x,y,z) := & (x < y \land y < z) \lor (z < y \land y < x).\\
\Cycl(x,y,z) := & (x < y \land y < z) \lor (y < z \land z < x) \lor (z < x \land x < y) \lor \\
 &(x < y \land z \bot xy) \lor (y < z \land x \bot yz) \lor (z < x \land y \bot zx).\\
\Par(x,y,z) := & (x \bot yz \land y \bot z) \lor (x < yz \land y \bot z) \lor (x  > yz \land y \bot z)\\
\Sept(x,y,z,t) := & (\Cycl(x,y,z) \land \Cycl(y,z,t) \land \Cycl(x,y,t) \land \Cycl(x,z,t)) \lor \\
& (\Cycl(z,y,x) \land \Cycl(t,z,y) \land \Cycl(t,y,x) \land \Cycl(t,z,x)).
\end{align*}
In Lemma \ref{lemma:relations} we are going to give a description of the monoids $\overline{\langle \updownarrow \rangle}$, $\overline{\langle \circlearrowright \rangle}$ and $\overline{\langle \updownarrow, \circlearrowright \rangle}$ as endomorphism monoids with the help of the above relations. We remark that $\Cycl$ and $\Par$ describes the orbits triples under $\langle \circlearrowright \rangle$ and $\Sept$ describes the orbit of a linearly ordered 4-tuple under $\langle \updownarrow, \circlearrowright \rangle$. 

\begin{lemma} \label{lemma:ppdefpar}
The incomparability relation $\bot$ is pp-definable in $(P; <, \Cycl)$ and $\Par$ is pp-definable in $(P;\Cycl)$.
\end{lemma}

\begin{proof}
%Observe first that $\Aut(P;\Cycl) = \langle \circlearrowright \rangle$ and $\Aut(P;<,\Cycl) = \Aut(\Po)$ by Theorem \ref{theorem:groups}. 
To proof the first part of the lemma, let
\begin{align*} \psi(x,y,a,b,c,d) := & \ x < a < c \land x < b < d \land y < c \land y < d  \land \Cycl(x,a,y) \land \Cycl(x,b,y) \\
& \land \Cycl(y,c,b) \land \Cycl(y,d,a) \land \Cycl(b,d,c) \land \Cycl(a,c,d).
\end{align*}

We claim that $x \bot y$ is equivalent to $\exists a,b,c,d~\psi(x,y,a,b,c,d)$. It is not hard to verify that $x \bot y$ implies $\exists a,b,c,d~\psi(x,y,a,b,c,d)$. For the other direction note that $\psi(x,y,a,b,c,d)$ implies that $x \neq y$ because $\Cycl(x,a,y)$ is part of the conjunction $\psi$.

Let us assume that $x < y$ and $\psi(x,y,a,b,c,d)$ holds for some elements $a,b,c,d \in P$. Then $\Cycl(x,a,y)$ implies that $a < y$, symmetrically we have $b < y$. Since $y < c,d$ we have that $a < d$ and $b < c$. Then $\Cycl(b,d,c)$ implies $d < c$ and $\Cycl(a,c,d)$ implies $c < d$, which is a contradiction.

Now assume that $y < x$ and $\psi(x,y,a,b,c,d)$ holds for some elements $a,b,c,d \in P$. Then we have $y < a,b$ by the transitivity of the order. Then $\Cycl(y,c,b)$ implies $c < b$ and $\Cycl(y,d,a)$ implies $d < a$. But this leads to the contradiction $a < c <b$ and $b < d < a$. 

For the second part of the lemma let $s,t \in P$ be two elements with $s < t$. Then the set $X = \{x \in P : s < x < t \}$ is pp-definable in $(P;\Cycl, s,t)$ by the formula $\phi(x) := \Cycl(s,x,t)$. By a back-and-forth argument one can show the two structures $(X;\leq)$ and $(P;\leq)$ are isomorphic. The order relation, restricted to $X$ is also pp-definable in $(P;\Cycl, s,t)$ by the equivalence
\[ y <_{|X} z \leftrightarrow \phi(x) \land \phi(y) \land \Cycl(y,z,t). \]

Since $\bot$ is pp-definable in $(P;<,\Cycl)$, we have that its restriction to $X$ has a pp-definition in $(P;\Cycl,s,t)$. 
Therefore also the relation $R = \{(x,y,z) \in X^3 : x \bot y \land  x \bot z \land  z \bot y \}$ is pp-definable in $(P;\Cycl,s,t)$. Let $\phi(s,t,u,v,w)$ be a primitive positive formula defining $R$.

We claim that $\exists x,y~\phi(x,y,u,v,w)$ is equivalent to $(u,v,w) \in \Par$. Let $(u,v,w) \in \Par$. The relation $\Par$ describes the orbit of a 3-element antichain under the action of $ \langle \circlearrowright \rangle \subseteq \End(P;\Cycl)$. So we can assume that $(u,v,w)$ is a 3-antichain, otherwise we take an image under a suitable function form $ \langle \circlearrowright \rangle$. Now let us take elements $s<t$ such that $s < uvw$ and $uvw < t$. Then clearly $\psi(s,t,u,v,w)$ has to hold.

Conversely let $(s,t,u,v,w)$ be a tuple such that $\psi(s,t,u,v,w)$ holds. We can assume that $s < t$ (otherwise we take the image of $(s,t,u,v,w)$ under a suitable function in $\langle \circlearrowright \rangle$). By what we proved above, $(u,v,w)$ is antichain, hence it satisfies $\Par$. \end{proof}

\begin{lemma} \ \label{lemma:relations}
\begin{enumerate}
\item $\End(P;<, \bot) = \overline{\Aut(\Po)}$
\item $\End(P; \Betw, \bot) = \overline{\langle \updownarrow \rangle}$
\item $\End(P;\Cycl) = \overline{\langle \circlearrowright \rangle}$
\item $\End(P;\Sept) = \overline{\langle \updownarrow, \circlearrowright \rangle}$
\end{enumerate}
\end{lemma}

\begin{proof} \
\begin{enumerate}
\item Clearly $\Aut(\Po) \subseteq \End(P;<, \bot)$. For the other inclusion let $ f \in \End(P;<, \bot)$.
Let $A \subseteq P$ be an arbitrary finite set. The restriction of $f$ to a finite subset $A \subseteq P$ is an isomorphism between posets. By the homogeneity of $\Po$ there is an automorphism $\alpha \in \Aut(\Po)$ such that $f \restriction A = \alpha \restriction A$.
\item Since $\updownarrow$ preserves $\Betw$ and $\bot$, we know that $\overline{\langle \updownarrow \rangle} \subseteq \End(P;\Betw, \bot)$ holds. For the opposite inclusion let $f \in \End(P;\Betw, \bot)$. If $f$ preserves $<$, then  $f \in \End(P;<, \bot)$ and we are done. Otherwise there is a pair of elements $c_1 < c_2$ with $f(c_1) > f(c_2)$. Let $d_1 < d_2$ be an other pair of points in $P$. Then there are $a_1,a_2 \in P$ such that $c_1 < c_2 < a_1< a_2 $ and $d_1 < d_2 <a_1 < a_2$. Since $f$ preserves $\Betw$, $f(a_1) > f(a_2)$ holds and hence also $f(d_1) > f(d_2)$. So $f$ inverts the order, while preserving $\bot$.Therefore $\updownarrow \circ f \in  \End(P;<, \bot)$. We conclude that $f \in \overline{\langle \updownarrow \rangle}$.
\item It is easy to see that $\overline{\langle \circlearrowright \rangle} \subseteq \End(P;\Cycl)$. So let $f \in \End(P;\Cycl)$. Clearly $f$ is injective and preserves also the relation $\Cycl'(x,y,z) := \Cycl(y,x,z)$. By Lemma \ref{lemma:ppdefpar}, $f$ also preserves the relation $\Par$. Furthermore $\overline{\langle \circlearrowright \rangle}$ is 2-transitive: This can be verified by the fact that for every two elements of $P$, we can find a $\alpha \in \Aut(\Po)$ that map one element to the random filter $F$ and the other element to $P \setminus F$. So also $\End(P;\Cycl)$ is 2-transitive. It follows that $\End(P;\Cycl)$ also preserves the negation of $\Cycl$. In other words, $f$ is a self-embedding of $(P;\Cycl)$. So, when restricted to a finite $A \subset P$, $f$ is a partial isomorphism. By the results in \cite{PPPS-a-new-transformation} we know that $(P;\Cycl)$ is a homogeneous structure. Hence for every finite $A \subset P$ we find an automorphism $\alpha \in \Aut(P;\Cycl) = \langle \circlearrowright \rangle$ such that $f\restriction A = \alpha \restriction A$. 
\item Let $f \in \End(P; \Sept)$. We claim that either $f$ or $\updownarrow \circ f$ preserves $\Cycl$. If we can prove our claim we are done by (3). First of all note that $\Sept(x,y,z,u)$ implies $\Cycl(x,y,z) \leftrightarrow \Cycl(y,z,u)$. 

Without loss of generality let there be a elements $x,y,z \in P$ with $\Cycl(x,y,z)$ and $\Cycl(f(x),f(y),f(z))$, otherwise we look at $\updownarrow \circ f$ instead of $f$. Let $(r,s,t)$ be arbitrary tuple satisfying $\Cycl$. 

We can always find elements $a < b < c$ in $P$ that are incomparable with all entries of $(x,y,z)$ and $(r,s,t)$. Further we can choose elements $u,v \in P$ that are incomparable with $(a,b,c)$ such that $z < u < v$ and $\Sept(x,y,z,u) \land \Sept(y,z,u,v)$ holds. This can be done by a case distinction and is left to the reader. By construction we have 
\[  \Sept(x,y,z,u) \land \Sept(y,z,u,v) \land \Sept(z,u,v,a) \land \Sept(u,v,a,b) \land \Sept(v,a,b,c). \] 
So we have that $(f(x),f(y),f(z)) \in \Cycl$ if and only if $(f(a),f(b),f(c)) \in \Cycl$. Repeating the same argument for $(r,s,t)$ gives us that $(f(r),f(s),f(t)) \in \Cycl$. So $f$ preserves $\Cycl$.
\end{enumerate}
\end{proof}

Recall that we obtain an ordered homogeneous Ramsey structure $(P;\leq,\prec)$ by taking the \fraisse{} limit of the class of finite structures $(A;\leq,\prec)$, where $(A;\leq)$ is a partial order on $A$ and $\prec$ an extension of $<$ to a total order. We can regard this structure to be an extension of $\Po$ by a total order. By Theorem \ref{theorem:ramsey} the following holds:

\begin{lemma}\label{lemma:blackbox} Let $f: P \to P$ and $c_1,\ldots, c_n \in P$ be any points. Then there exists a function $g:P \to P$ such that
\begin{enumerate}
\item $g \in \overline{ \langle \Aut(\Po) \cup \{f\} \rangle }$.
\item $g(c_i)=f(c_i)$ for $i=1,\ldots n$.
\item Regarded as a function from $(P;\leq,\prec,\bar{c})$ to $(P;\leq)$, $g$ is a canonical function. \hfill $\square$\end{enumerate} 
\end{lemma}

Let $\Gamma$ be a reduct of $\Po$. We are going to study all feasible behaviors of a canonical function $f : (P;\leq,\prec,\bar{c}) \to (P;\leq)$ when $f \in \End(\Gamma)$. Note that the behaviour of such $f$ only depends on the behaviour on the 2-types because $(P;\leq,\prec, \bar c)$ is homogeneous and its signature contains at most 2-ary relation symbols. Since there are only finitely many 2-types, the study of all possible behaviors of such canonical functions is a combinatorial problem.
We introduce the following notation:

\begin{notation} \label{notation:canonical}
Let $A,B$ be definable subsets of $\Po$ and let $\phi_1(x,y),\ldots,\phi_n(x,y)$ be formulas. We let $p_{A,B,\phi_1,\ldots,\phi_n}(x,y)$ denote the (partial) type determined by the formula $x \in A \wedge y \in B \wedge \phi_1(x,y) \wedge \ldots \wedge \phi_n(x,y)$. Using this notation, we can describe the 2-types of $(P;\leq,\prec,\bar{c})$. They are all of the form $p_{X,Y,\phi, \psi}= \{(a,b) \in P^2: a \in X, b \in Y, \phi(a,b)$ and $\psi(a,b)\}$, where $X$ and $Y$ are 1-types, $\phi \in \{=,<,>,\bot\}$ and $\psi \in \{=,\prec,\succ\}$.

Let $X,Y$ be two distinct infinite 1-types of $(P;\leq,\prec,\bar c)$. We write $X \lowbot Y$ if there are pairs $(x,y),(x',y') \in X \times Y$ with $x < y$ and $x' \bot y'$.

When it is convenient for us we will abuse notation and write $\bar c$ to describe the set containing all entries of the tuple $\bar c$.
\end{notation}
\begin{observation} \label{observation:types}
The structure $(P;\leq,\prec,\bar c)$ is a homogeneous structure. If $X$ is an 1-type of $(P;\leq,\prec,\bar c)$ with infinite elements, then $(X;\leq, \prec)$ is isomorphic to $(P;\leq, \prec)$ itself. This can be seen by a back-and-forth argument. Similarly, if $X$ and $Y$ are 1-types of $(P;\leq,\prec,\bar c)$ with infinite elements such that $X \lowbot Y$ holds, then $X\cup Y$ is isomorphic to $(P;\leq)$ with $X$ being a random filter. If we define $X \leq Y \leftrightarrow \exists (x,y) \in X\times Y~(x \leq y)$ we get a partial order on the 1-types of $(P;\leq,\prec,\bar c)$ (confer Lemma 18 of \cite{Poset-Reducts}).
But note that the 1-types of $(P;\leq,\prec,\bar c)$ are not necessarily linearly ordered by $\prec$: There can be infinite 1-types $X,Y$ and $(x,y),(x',y') \in X \times Y$ with $x \prec y$, $x \bot y$ and $y' \prec x'$, $x' \bot y'$.
\end{observation}

In the following lemmas let $\Gamma$ be always be a reduct of $\Po$ and let $f \in \End(\Gamma)$ be a canonical function from $(P;\leq,\prec,\bar c)$ to $(P;\leq)$.

\begin{lemma} \label{lemma:singleorbit}
Let $X$ be a 1-type of $(P;\leq,\prec,\bar c)$ with infinite elements. Then $f$ behaves like $id$ or $\updownarrow$ on $X$, otherwise $\End(\Gamma)$ contains a constant function, $g_{<}$ or $g_{\bot}$.
\end{lemma}

\begin{proof}
Note that $(X;\leq,\prec)$ is isomorphic to $(P;\leq,\prec)$. Then we can prove the statement with the same arguments as in Lemma 8 of \cite{Poset-Reducts}.
%Assume $f$ does not behave like like $id$ or $\updownarrow$ on $X$.
%
%It can be easily seen that if $f$ is non-injective on $X$ then $f$ has to be constant on $X$.
%
%If $f(p_{X,<, \prec}) = p_{\bot}$ and $f(p_{X,\bot, \prec}) = p_{\bot}$, then $g_{\bot}$ can be obtained by composing a homomorphism from $(P;\leq,\prec)$ to $X$ with $f$.
%
%If $f(p_{X,<, \prec}) = p_{<}$ and $f(p_{X,\bot, \prec}) = p_{<}$, then $g_{<}$ can be obtained by composing a homomorphism from $(P;\leq,\prec)$ to $X$ with $f$.
%
%If $f(p_{X,<, \prec}) = p_{>}$ and $f(p_{X,\bot, \prec}) = p_{>}$, then $g_{<}$ can be obtained by first mapping $(P;\leq,\prec)$ to $X$, then applying $f$, then send the image again to $X$ and applying $f$ again. For short, we say $g_{<}$ can be obtained by applying $f$ twice.
%
%It is easy to see that all remaining cases are not possible, since they would induce structures that are not partial orders. For such examples consider a finite $A = \{a_1,a_2,a_3 \}$ with $a_1 \prec a_2 \prec a_3$, $a_1 < a_3$ and $a_1 \bot a_2$, $a_2 \bot a_3$. (cf \cite[Lemma 8]{Poset-Reducts}).
\end{proof}

\begin{lemma} \label{lemma:onXY}
Let $X,Y$ two infinite 1-types of $(P;\leq,\prec,\bar c)$ with $X \lowbot Y$. Assume $f$ behaves like $id$ on $X$. Then $f$ behaves like $id$ or $\circlearrowright_X$ on $X \cup Y$, otherwise $\End(\Gamma)$ contains a constant function, $g_{<}$ or $g_{\bot}$.
\end{lemma}

\begin{proof}
Assume that $f$ does not contains a constant function, $g_{<}$ or $g_{\bot}$. Note that the union of $X$ and $Y$ is isomorphic to $\Po$ and $X$ is a random filter of $X \cup Y$. By following the arguments of Lemma 22 in \cite{Poset-Reducts} one can show that we only have the two possibilities that
\begin{enumerate}
\item $f(p_{X,Y,<}) = p_{<}$ and $f(p_{X,Y,\bot,\prec}) = p_{\bot}$ or
\item $f(p_{X,Y,<}) = p_{\bot}$ and $f(p_{X,Y,\bot,\prec}) = p_{>}$.
\end{enumerate}
By Lemma \ref{lemma:singleorbit} we may assume that $f$ behaves like $id$ or $\updownarrow$ on $Y$. But if $f$ behaves like $\updownarrow$ on $Y$, the image of $y_1, y_2 \in Y$ and $x \in X$ with $x \prec y_1 < y_2$, $x \bot y_1$ and $x < y_2$ would be a non partially ordered set. So if the type $p_{X,Y,\bot,\succ}$ is empty, $f$ behaves like $id$ or $\circlearrowright_X$ on $X \cup Y$ and we are done.

If $p_{X,Y,\bot,\succ}$ is not empty, there are $x \in X$ and $y \in Y$ with $x \succ y$ and $x \bot y$. We claim that in this case $f(p_{X,Y,\bot,\succ}) = f(p_{X,Y,\bot,\prec})$. We only prove this claim for (1), the proof for (2) is the same. 

Assume that $f(p_{X,Y,\bot,\succ}) = p_{<}$. Then let $x' \in X$ be an element such that $y \prec x'$ and $x < x'$ and $y \bot x'$. The fact that such an element exists can be verified by checking that the extension of $\{x,y\} \cup \bar c$ by such an element $x'$ still lies in the age of $(P;\leq,\prec,\bar c)$. By our assumption we then have $f(x) < f(x') < f(y)$, which contradicts to $f(x) \bot f(y)$.

Now assume that $f(p_{X,Y,\bot,\succ}) = p_{>}$. Then let $x' \in X$ be such that $x \prec y \prec x'$ and $x < y$ and $x' \bot x y$. Again the fact that $x'$ exists can be verified by the homogeneity of $(P;\leq,\prec,\bar c)$. Then $f(x) < f(y) < f(x')$, which contradicts to $f(x') \bot f(x')$.
\end{proof}

\begin{lemma} \label{lemma:everyorbit}
Either $f$ behaves like $id$ or $\updownarrow$ on every single 1-type or $\End(\Gamma)$ contains a constant function, $g_{<}$ or $g_{\bot}$.
\end{lemma}

\begin{proof}
For every two infinite orbits $X < Y$ there is a infinite orbit $Z$ with $X \lowbot Z$ and $Z \lowbot Y$. For every two infinite orbits $X \bot Y$ there is an infinite orbit $Z$ with $X < Z$ and $Y < Z$. So this statement holds by Lemma \ref{lemma:onXY}.
(cf Lemma 23 of\cite{Poset-Reducts})
\end{proof}

\begin{lemma}  \label{lemma:allorbits}
Assume $\End(\Gamma)$ does not contains constant functions, $g_{<}$ or $g_{\bot}$. Then there is a $g \in \overline{\langle \circlearrowright, \updownarrow \rangle} \cap \End(\Gamma)$ such that $g \circ f$ is canonical from $(P;\leq,\prec,\bar c)$ to $(P;\leq)$ and behaves like $id$ on every set $(P \setminus \bar c ) \cup \{c\}$, with $c \in \bar c$. 
\end{lemma}

\begin{proof}
By Lemma \ref{lemma:everyorbit}, $f$ behaves like $id$ or $\updownarrow$ on every infinite orbit. Without loss of generality we can assume that the first case holds, otherwise consider $\updownarrow \circ f$. 

Let $X \lowbot Y$, $Y \lowbot Z$ and $X \lowbot Z$ or $X < Z$. If $f$ behaves like $id$ on $X \cup Y$ and $Y \cup Z$ it also has to behave like $id$ on $X \cup Z$; otherwise the image of a triple $(x,y,z)\in X \times Y\times Z$ with $x<y<z$ would not be partially ordered. Let $X < Z$, $Y < Z$ and $X \bot Z$. Again, if $f$ behaves like $id$ on $X \cup Y$ and $Y \cup Z$ it also has to behave like $id$ on $X \cup Z$, otherwise we get a contradiction.

By Lemma \ref{lemma:onXY} $f$ either behaves like $id$ or like $\circlearrowright_X$ on the union two orbits $X \lowbot Y$. In the second case $\circlearrowright \, \in \End(\Gamma)$. The set $A = \{x \in P: y < x \lor y \bot x \text{ for all } y \in f(Y) \}$ is a union of orbits of $\Aut(P;\leq,\prec,\bar c)$ and a random filter of $P$. So $\circlearrowright_{A} \circ f$ is canonical and behaves like $id$ on $X \cup Y$. Repeating this step finitely many times gives us a function $g \in \langle \circlearrowright \rangle$ such that $g \circ f$ behaves like $id$ on the union of infinite orbits, by the observations in the paragraph above.

It is only left to show that $g \circ f$ behaves like $id$ between a given constant $c$ in $\bar c$ and an infinite orbit $X$. Assume for example that $c < X$ and $g \circ f(p_{c,X,<}) = p_{\bot}$. Let $A \subseteq P$ with $a \in A$. By homogeneity of $\Po$ we find an automorphism of $\Po$ that maps $a$ to $c$ and all points that are greater than $a$ to $X$. If we then apply $g \circ f$ and repeat this process at most $|A|$-times we can map $A$ to an antichain. Thus $g_{\bot} \in \End(\Gamma)$ which contradicts to our assumption.

Similarly all other cases where $g \circ f$ does not behave like $id$ between $c$ and $X$ contradict our assumptions. We leave the proof to the reader. Hence $g \circ f$ behaves like $id$ everywhere except on $\bar c$.
\end{proof}

Now we are ready to proof the main result of the section.

\begin{proof}[Proof of Proposition \ref{theorem:monoids}]

Let $\Gamma$ be a reduct of $\Po$ such that $\End(\Gamma)$ does not contains constant functions, $g_{<}$ or $g_{\bot}$. We show that then $\End(\Gamma)$ is equal to $\overline{\Aut(\Po)}$, $\overline{\langle \updownarrow \rangle}$, $\overline{\langle \circlearrowright \rangle}$ or $\overline{\langle \updownarrow, \circlearrowright \rangle}$.

First assume that $\End(\Gamma)$ contains a non injective function. This can be witnessed by constants $c_1 \neq c_2$ and a function $f \in \End(\Gamma)$ with $f(c_1) = f(c_2)$ that is canonical as function $f: (P;\leq,\prec, c_1, c_2) \to (P; \leq)$. By Lemma \ref{lemma:allorbits} we can assume that $f$ behaves like $id$ everywhere except from $c_1,c_2$. But this is not possible, since there is a point in $a \in P$ with $a \bot c_1$ but $\neg(a \bot c_2)$. Since $f(c_1) = f(c_2)$ either $<$ or $\bot$ is violated, which contradicts to $f$ behaving like $id$ everywhere except on $\{c_1, c_2\}$. So from now on let $\End(\Gamma)$ only contain injective functions.

Assume $\End(\Gamma)$ violates $\Sept$. This can also be witnessed by a canonical function $f: (P;\leq,\prec, \bar c) \to (P; \leq)$ such that $\bar c \in \Sept$ but $f(\bar c) \not\in \Sept$. By Lemma \ref{lemma:allorbits} we can assume that $f$ behaves like $id$ on every set $(P \setminus \bar c ) \cup \{c\}$, with $c \in \bar c$. If there are $c_i < c_j$ with $f(c_i) \bot f(c_j)$ it is easy to see that $\End(\Gamma)$ generates $g_{\bot}$ which contradicts to our assumptions. If there are $c_i < c_j$ or $c_i \bot c_j$ with $f(c_i) > f(c_j)$ let $a$ be an element of $(P \setminus \bar c )$ with $a < c_j$ and $a \bot c_i$. Then the image of $a,c_i,c_j$ under $f$ induces a non partially ordered structure - contradiction.

So $\End(\Gamma)$ preserves $\Sept$. By Lemma \ref{lemma:relations} we know that $\End(\Gamma) \subseteq \overline{\langle \updownarrow, \circlearrowright \rangle}$. If $\End(\Gamma)$ violates $\Cycl$ and $\Betw$ or $\Cycl$ and $\bot$ we can proof as in the paragraph above that $\End(\Gamma) = \overline{\langle \updownarrow, \circlearrowright \rangle}$.

Similarly, if $\End(\Gamma)$ preserves $\Cycl$ but violates $\Betw$ or $\bot$ then $\End(\Gamma) =  \overline{\langle \circlearrowright \rangle}$.

If $\End(\Gamma)$ preserves $\Betw$ and $\bot$ but violates $\Cycl$. Then $\End(\Gamma) =  \overline{\langle \updownarrow \rangle}$.

Finally, if $\End(\Gamma)$ preserves $\Betw$, $\bot$ and $\Cycl$ we have $\End(\Gamma) = \overline{\Aut(\Po)}$.
\end{proof}

%Throughout the remaining parts of this paper we are going to study the complexity of $\csp(\Gamma)$ for model-complete reducts $\Gamma$ of $\Po$. A first observation regards the case where $\End(\Gamma) = \overline{\langle \updownarrow \rangle}$.

%\begin{lemma} \label{lemma:betwhard}
%Let $\Gamma$ be a reduct of $\Po$ such that $\End(\Gamma) = \overline{\langle \updownarrow \rangle}$ and $u \in P$. Then there is a pp-interpretation of $(\{0,1\}; \rm 1IN3 )$ in $(\Gamma,u)$ and $\csp(\Gamma)$ is NP-complete.
%\end{lemma} 
%
%\begin{proof}
%Note that the betweenness relation $\Betw$ is an orbit of $\End(\Gamma) = \overline{\langle \updownarrow \rangle}$ on $P^3$. Now Theorem \ref{theorem:violate} implies that $\Betw$ is primitively positive definable in $\Gamma$. We can show that $(\Q;\Betw)$ has a pp-interpretation in $\Gamma$; in fact every bijective map $I: P \to \Q$ preserving $<$ is such an interpretation. It is known $(\{0,1\}; \rm 1IN3 )$ has a pp interpretation in $(\Q,\Betw,0)$; one can find a proof in Proposition 5.5.13. of \cite{manuelhabil}. Now let $u$ be the preimage of $0$ under $e$. Then, by the transitivity of pp-interpretation we know that $(\{0,1\}; \rm 1IN3 )$ has a pp-interpretation in $(\Gamma,u)$.\\
%By Theorem \ref{theorem:coreconst} we have that $\csp(\Gamma)$ is NP-complete.
%\end{proof}

\section{The case where $<$ and $\bot$ are pp-definable} \label{sect:lowbotpp}

Throughout the remaining parts of this paper we are going to study the complexity of $\csp(\Gamma)$ for model-complete reducts $\Gamma$ of $\Po$. We start with the case where $\End(\Gamma)$ is the topological closure of the automorphism group of $\Po$. In this case the two relations $<$ and $\bot$ are pp-definable by Theorem \ref{theorem:violate}. So throughout this section let $\Gamma$ be a reduct of $\Po$ in which $<$ and $\bot$ are pp-definable. We are first going to discuss the binary part of the $\Pol(\Gamma)$. This will be essential for proving the dichotomy in this case.

\begin{observation} \label{obs:lowbot}
The binary relation $x \lowbot y$ defined by $x < y \lor x \bot y$ is equivalent to the primitive positive formula $\exists z~(z < y) \land z \bot x$. Therefore $x \lowbot y$ is pp definable in $\Gamma$.
\end{observation}

By $e_<$ we denote an embedding of the structure $(P;<)^2$ into $(P;<)$. Clearly $e_<$ is canonical when regarded as map $e_<: (P;\leq, \prec)^2 \to (P;\leq)$. It has the following behaviour:
\[
\begin{array}{c|cccc}
e_< &    =  &   <   &    >  & \bot \\
\hline
= &    =  & \bot & \bot & \bot \\
< & \bot &    <  & \bot &  \bot \\
> & \bot & \bot &   >   &  \bot \\
\bot  & \bot & \bot & \bot  &  \bot
\end{array}
\]

By $e_{\leq}$ we denote an embedding of $(P;\leq)^2$ into $(P;\leq)$ that is canonical function when regarded as map $e_{\leq}: (P;\leq, \prec)^2 \to (P;\leq)$. It has the following behaviour:

\[
\begin{array}{c|cccc}
e_{\leq} &    =  &   <   &    >  & \bot \\
\hline
= & =  & < & > & \bot \\
< & < &    <  & \bot &  \bot \\
> & > & \bot &   >   &  \bot \\
\bot  & \bot & \bot & \bot  &  \bot
\end{array}
\]

\subsection{Horn tractable CSPs given by $e_<$ and $e_{\leq}$} \label{sect:horn}

The two functions $e_<$ or $e_{\leq}$ are of central interest to us. We will show in this section that if one of them is a polymorphisms of $\Gamma$, then the problem $\csp(\Gamma)$ is tractable. %In both cases $\csp(\Gamma)$ belongs to the class of Horn-tractable problems described in \cite{BCK-maximalCSP}.

Let $\Delta$ and $\Lambda$ be relational structures of the same signature. We say a map $h: \Delta \to \Lambda$ is a \textit{strong homomorphism} if $\bar x \in R \leftrightarrow h(\bar x) \in R$. By $\hat \Delta$ we denote the extension of $\Delta$ that contains the negation $\neg R$ for every $R$ is in $\Delta$.

\begin{theorem}[Proposition 14 from \cite{BCK-maximalCSP}] \label{theorem:horn}
Let $\Delta$ be an $\omega$-categorical structure and let $\Gamma$ be a reduct of $\Delta$. Suppose $\csp(\hat \Delta)$ is tractable. If $\Gamma$ has a polymorphism that is a strong homomorphism from $\Delta^2$ to $\Delta$, then also $\Gamma$ is tractable. \hfill $\square$
\end{theorem}

By definition $e_<$ is a strong homomorphism from $(P;<)^2 \to (P;<)$ and $e_{\leq}$ s a strong homomorphism from $(P;\leq)^2 \to (P;\leq)$. Let $\not <$ respectively $\not\leq$ denote the negation of the order relation $<$ respectively $\leq$. One can see that every input to $\csp(P;<,\not<)$ and $\csp(P;\leq,\not\leq)$ is accepted as long as it does not contradict to the transitivity of $<$ respectively $\leq$. But this can be checked in polynomial time, thus the two problems are tractable. So by Theorem \ref{theorem:horn} every template $\Gamma$ with polymorphism $e_<$ or $e_\leq$ gives us a tractable problem.

In the following theorem we additionally give a semantic characterization of these tractable problems via Horn formulas. This characterisation works also in the general setting, we refer to \cite{BCK-maximalCSP} for the proof.
%A semantic characterization of these tractable problems can be given by Horn formulas. We say that a relation $R$ has a quantifier-free \textit{Horn definition} in $\Delta$ if $R$ can be defined by a quantifier-free first-order formula over the signature of $\Delta$ together with $=$, that is in conjunctive normal form in which every clause contains at most one positive literal.

\begin{theorem} \label{theorem:horn2}
Let $\Gamma$ be a reduct of $\Po$. Suppose that $e_\leq \in \Pol(\Gamma)$. Then $\csp(\Gamma)$ is tractable and every relation in $\Gamma$ is equivalent to Horn formula in $(P;\leq)$:
\begin{align*}
x_{i_1} \leq x_{j_1} \land x_{i_2} \leq x_{j_2} \land \cdots \land x_{i_k} \leq x_{j_k} &\to x_{i_{k+1}} \leq x_{j_{k+1}} \text{ or} \\
x_{i_1} \leq x_{j_1} \land x_{i_2} \leq x_{j_2} \land \cdots \land x_{i_k} \leq x_{j_k} &\to \text{ 'false'}
\end{align*}
Suppose that $e_< \in \Pol(\Gamma)$. Then $\csp(\Gamma)$ is tractable and every relation in $\Gamma$ is equivalent to a Horn formula in $(P;<)$, i.e. a formula of the form:
\begin{align*}
x_{i_1} \vartriangleleft_1 x_{j_1} \land x_{i_2} \vartriangleleft_2 x_{j_2} \land \cdots \land x_{i_k} \vartriangleleft_k x_{j_k} &\to x_{i_{k+1}} \vartriangleleft_{k+1} x_{j_{k+1}} \text{ or} \\
x_{i_1} \vartriangleleft_1 x_{j_1} \land x_{i_2} \vartriangleleft_2 x_{j_2} \land \cdots \land x_{i_k} \vartriangleleft_k x_{j_k} &\to \text{ 'false'},
\end{align*}
where $\vartriangleleft_i~\in \{<,=\}$ for all $i = 1, \ldots, k+1$.\hfill $\square$
\end{theorem}

\subsection{Canonical binary functions on $(P;\leq,\prec)$}
A first step in analysing the binary part of $\Pol(\Gamma)$ is to look at the special case of canonical functions. So in the following text we are going to study the behaviour of binary functions $f \in \Pol(\Gamma)$ that are canonical seen as functions from $(P;\leq,\prec)^2$ to $(P;\leq)$. We are going to specify conditions for which $\Pol(\Gamma)$ contains $e_<$ or $e_\leq$.

\begin{definition} \label{def:dominated}
Let $f: \Po^2 \to \Po$ be a function. Then $f$ is called \textnormal{dominated on the first argument} if
\begin{itemize}
\item $f(x,y) < f(x',y')$ for all $x < x'$ and
\item $f(x,y) \bot f(x',y')$ for all $x \bot x'$.
\end{itemize}
We say $f$ is \textnormal{dominated} if $f$ or $(x,y) \mapsto f(y,x)$ is dominated on the first argument.
\end{definition}

We are going to prove the following theorem:

\begin{theorem} \label{theorem:binarycanonical}
Let $\Gamma$ be a reduct of $\Po$ in which $<$ and $\bot$ are pp-definable. Let $f(x,y) \in \Pol(\Gamma)$ be canonical when seen as a function from $(P;\leq,\prec)^2$ to $(P;\leq)$. Then at least one of the following cases holds:
\begin{itemize}
\item $f$ is dominated
\item $\Pol(\Gamma)$ contains $e_<$
\item $\Pol(\Gamma)$ contains $e_{\leq}$
\end{itemize}
\end{theorem}

First of all we make some general observations for binary canonical functions preserving $<$ and $\bot$. We are again going to use the notation introduced in Notation \ref{notation:canonical}. Let us fix a function $-:(P;\leq,\prec) \to (P;\leq,\prec)$ such that $x \prec y \leftrightarrow - y \prec - x$ holds. It is easy to see that such a function exists. 

\begin{lemma} \label{lemma:behaviour1}
Let $f: (P;\leq,\prec)^2 \to (P;\leq)$ be canonical and $f \in \Pol(\Gamma)$. Then the following statements are true:
\begin{enumerate}
\item $f(p_<,p_<) = p_<$, $f(p_\bot,p_\bot) = p_\bot$
\item $f(p,q) = - f(-p,-q)$, for all types $p,q$.
\item $f(p_<,p_{\bot, \prec})$, $f(p_<,p_{\bot, \succ})$,  $f(p_{\bot, \succ}, p_<)$ and $f(p_{\bot, \prec}, p_<)$ can only be equal to $p_<$ or $p_\bot$.
\item At least one of $f(p_<,p_{\bot, \prec})$ and $f(p_{\bot, \prec},p_<)$ is equal to $p_\bot$.
\item At least one of $f(p_<,p_{\bot, \succ})$ and $f(p_{\bot, \prec},p_>)$ is equal to $p_\bot$.
\item It is not possible that $f(p_<,p_>) = p_=$ holds.
\item $f(p_{\bot,\prec},p_<) = p_\bot \to f(p_{\bot},p_=) = p_\bot$
%\item $f(p_{\bot, \prec},p_=) = p_< \to f(p_{<},p_=) = p_{<}$ and $f(p_{\bot, \prec},p_=) = p_> \to f(p_{<},p_=) = p_{>}$.
%\item $f(p_{\bot, \prec},p_=) = p_=$ if and only if $f(p_{<},p_=) = p_{=}$
\end{enumerate}
\end{lemma}

\begin{proof} \
\begin{enumerate}
\item This is clear, since $f$ is a polymorphism of $\Gamma$ and hence preserves $<$ and $\bot$.
\item This is true by definition of $-$.
\item This is true since $f$ preserves the relation $\lowbot$, see Observation \ref{obs:lowbot}.
%\item Assume $f(p_<,p_{\bot, \prec}) = p_>$. Let $a_1 < a_2 < a_3$ and $b_1 \prec b_2 \prec b_3$, $b_1 \bot b_2$, $b_2 \bot b_3$, $b_1 < b_3$. Then $f(a_1,b_1) < f(a_3,b_3)$ and $f(a_1,b_1) > f(a_2,b_2) > f(a_3,b_3)$ have to hold, which is a contradiction. Similarly $f(p_<,p_{\bot, \prec}) = p_=$ leads to a contradiction. For the result for $f(p_{\bot, \prec},p_<)$ consider the function $(x,y) \to f(y,x)$.
%\begin{center}
%\begin{tikzcd}[column sep = small, row sep = small]
%a_3 						& b_3 				         & \arrow[d] f(a_3,b_3) 			\\
%\arrow[u] a_2 				& \arrow[u, dashed] b_2 		 & \arrow[d] f(a_2,b_2) 			\\
%\arrow[uu, bend right]\arrow[u]  a_1 & \arrow[u, dashed]  \arrow[uu, bend right] b_1 & \arrow[uu, bend right] f(a_1,b_1) 
%\end{tikzcd}
%\end{center}
\item Assume $f(p_<,p_{\bot, \prec}) = f(p_{\bot, \prec},p_<) = p_<$. Let $a_1 \prec a_2 \prec a_3$ with $a_1 < a_2$, $a_3 \bot a_1 a_2$ and $b_1 \prec b_2 \prec b_3$ with $b_2 < b_3$, $b_1 \bot b_2 b_3$. By our assumption  $f(a_1,b_1) < f(a_2,b_2) < f(a_3,b_3)$ holds, which is a contradiction to $f$ preserving $\bot$.
%\begin{center}
%\begin{tikzcd}[column sep = small, row sep = small]
%a_3 						& b_3 				         & f(a_3,b_3) 			\\
%\arrow[u, dashed] a_2 				& \arrow[u] b_2 		 & \arrow[u] f(a_2,b_2) 			\\
%\arrow[uu, bend right, dashed]\arrow[u]  a_1 & \arrow[u, dashed]  \arrow[uu, bend right, dashed] b_1 & \arrow[uu, bend right, dashed]\arrow[u] f(a_1,b_1) 
%\end{tikzcd}
%\end{center}
\item This can be proven similarly to (4).
\item Assume that $f(p_<,p_>) = p_=$ holds. Let $a_1 \prec a_2 \prec a_3$ with $a_1< a_3$, $a_2 < a_3$, $a_1 \bot a_2$ and $b_1 \succ b_2 \succ b_3$ with $b_1 > b_3$, $b_2 > b_3$, $b_1 \bot b_2$. Then $f(a_1,b_1) \bot f(a_2,b_2)$ but also $f(a_1,b_1) = f(a_3,b_3) = f(a_2,a_2)$ have to hold, which is a contradiction.

\item Assume that there are $a_1 \bot a_2$ and $b$ such that $f(a_1,b) \leq f(a_2,b)$ holds. Then we take elements $a_3$ and $b'$ with $a_2 < a_3$, $a_1 \bot a_3$ , $a_1 \prec a_3$ and $b' > b$. Then $f(a_1,b) \leq f(a_2,b) < f(a_3,b')$ holds, which is a contradiction to $f(a_1,b) \bot f(a_3, b')$.
%\item Let $a_1 \prec a_2 \prec a_3$ with $a_1 < a_3$ and $a_2 \bot a_1 a_3$. Then $f(a_1,b) < f(a_2,b) < f(a_3,b)$ holds. This  implies $f(a_1,b) < f(a_3,b)$, thus $f(p_{<},p_=) = p_{<}$. The other statement can be proven symmetrically.
%\item Assume $f(p_{\bot, \prec},p_=) = p_=$. Let $a_1 \prec a_2 \prec a_3$ with $a_1 < a_2$ and $a_2 \bot a_1 a_3$. Then $f(a_1,b) = f(a_2,b) = f(a_3,b)$ holds. This  implies $f(a_1,b) < f(a_3,b)$, thus $f(p_{<},p_=) = p_{<}$. For the converse assume $f(p_<,p_=) = p_=$. Let $a_1 < a_2 a_3$ with $a_2 \bot a_3$. Then $f(a_1,b) = f(a_2,b)$ and $f(a_1,b) = f(a_3,b)$ implies that $f(p_{\bot, \prec},p_=) = p_=$.
\end{enumerate}
\end{proof}

By Lemma \ref{lemma:behaviour1} (2) we only have to consider pairs of types where the first entry is $p_{=}$, $p_<$ or $p_{\bot,\prec}$ when studying the behaviour of $f$. Further Lemma \ref{lemma:behaviour1} implies that $f(x,y) \neq f(x',y')$ always holds for $x \neq x'$ and $y \neq y'$.

\begin{lemma} \label{lemma:dominated}
Let $f \in \Pol(\Gamma)$. Then the following are equivalent:
\begin{enumerate}
\item $f(p_<,p_>) = p_<$
\item $f(p_<,p_{\bot, \succ}) = p_<$
\item $f(p_<,q) = p_<$ for all 2-types $q$
\item $f$ is dominated in the first argument
\end{enumerate}
\end{lemma}

\begin{proof}
It is clear that the implications (4) $\to$ (3) $\to$ (2) and $(3)$ $\to$ (1) are true.

(1) $\to$ (3):
Let $a_1 < a_2 < a_3$ and $b_1 b_3 < b_2$. Then $f(a_1,b_1) < f(a_2,b_2) < f(a_3,b_3)$ has to hold regardless if the type of $(b_1,b_3)$ is $p_{\bot, \prec}$, $p_{\bot, \succ}$ or $p_=$. So $f(p_<,q) = p_<$ for all 2-types $q$.
%
%\begin{center}
%\begin{tikzcd}[column sep = small, row sep = small]
%a_3 						& b_3 \arrow[d]				         & f(a_3,b_3) 			\\
%\arrow[u] a_2 				& b_2 		 & \arrow[u] f(a_2,b_2) 			\\
%\arrow[uu, bend right]\arrow[u]  a_1 & \arrow[u] b_1 & \arrow[uu, bend right]\arrow[u] f(a_1,b_1) 
%\end{tikzcd}
%\end{center}

(2) $\to$ (1):
Let $a_1 < a_2 < a_3$ and $b_1 \succ b_2 \succ b_3$ with $b_1 > b_3$, $b_2 \bot b_1 b_3$. Then $f(a_1,b_1) < f(a_2,b_2) < f(a_3,b_3)$ implies $f(a_1,b_1) < f(a_3,b_3)$ and so $f(p_<,p_>) = p_<$.
%
%\begin{center}
%\begin{tikzcd}[column sep = small, row sep = small]
%a_3 						& b_3 \arrow[dd, bend left]\arrow[d, dashed]	         & f(a_3,b_3) 			\\
%\arrow[u] a_2 				& b_2 \arrow[d, dashed]	 & \arrow[u] f(a_2,b_2) 			\\
%\arrow[uu, bend right]\arrow[u]  a_1 & b_1 & \arrow[uu, bend right]\arrow[u] f(a_1,b_1) 
%\end{tikzcd}
%\end{center}

(3) $\to$ (4): We have to consider all the pairs of 2-types where the first entry is $p_{\bot,\prec}$. By Lemma \ref{lemma:behaviour1} (4) and (5) we know that $f(p_{\bot,\prec},p_<) = f(p_{\bot,\prec},p_>) = p_{\bot}$. From Lemma \ref{lemma:behaviour1}(7) follows that $f(p_\bot,p_=) = p_\bot$. 
%So it remains to show that $f(p_{\bot,\prec},p_{>}) = p_{\bot}$. The other option $f(p_{\bot,\prec},p_{>}) = p_{>}$ is not possible: Let $a_1 \prec a_2 \prec a_3$ with $a_2 \bot a_1 a_3$ and $a_1 < a_3$. Let $b_3 \prec b_2 \prec b_1$ with $b_2 < b_1$ and $b_3 \bot b_1 b_2$. Then $f(a_2,b_2) < f(a_1,b_1) < f(a_3,b_3)$, which is a contradiction to $f$ preserving $\bot$.
%\begin{center}
%\begin{tikzcd}[column sep = small, row sep = small]
%a_3 						& b_3 \arrow[dd, bend left]\arrow[d, dashed]	         & f(a_3,b_3) \arrow[dd, bend left]\\
% a_2\arrow[u, dashed]			& b_2	 & \arrow[u,dashed]\arrow[u, dashed] f(a_2,b_2) 			\\
%\arrow[u]\arrow[uu, bend right, dashed]  a_1 & b_1\arrow[u, dashed]	 & \arrow[u] f(a_1,b_1) 
%\end{tikzcd}
%\end{center}

We want to point out that we did not require $f$ to be canonical; it can be easily verified that all proof steps also work for general binary functions.
\end{proof}

\begin{lemma} \label{lemma:notdominated}
Let $f: (P;\leq,\prec)^2 \to (P;\leq)$ be canonical and $f \in \Pol(\Gamma)$. If $f$ is not dominated the following statements are true:
\begin{enumerate}
\item $f(p_<,p_>) = f(p_<,p_{\bot, \succ}) = f(p_{\bot, \prec},p_>) = p_\bot$.
\item $f(p_<,p_=) = p_<$ or $f(p_<,p_=) = p_\bot$.
\item $f(p_{\bot, \prec},p_=) = p_{\bot}$ or $f(p_{\bot, \prec},p_=) = p_{<}$.
%\item $f(p_{\bot, \prec},p_<) = p_{\bot}$ implies $f(p_{\bot},p_=) = p_{\bot}$.
\end{enumerate}
\end{lemma}

\begin{proof} \
\begin{enumerate}
\item is a direct consequence of Lemma \ref{lemma:dominated}.
\item Suppose there are $a_1 < a_2$ and $b$ such that $f(a_1,b) \geq f(a_2,b)$. Then we take elements $a_3,b' \in P$ with $a_2 \bot a_3$ $a_2 \succ a_3$, $a_1 < a_3$ and a $b' > b$. Then $f(a_2,b) \leq f(a_1,b) <f(a_3,b')$ holds, which is a contradiction to $f(a_2,b) \bot f(a_3,b')$.

\item Assume that there are $a_1 \bot a_2$, $a_1 \prec a_2$ and $b$ such that $f(a_1,b) \geq f(a_2,b)$ holds. There are elements $a_3$ and $b'$ with $a_2 > a_3$, $a_1 \bot a_3$, $a_1 \prec a_3$ and $b' < b$. Then $f(a_2,b) > f(a_3,b')$ and $f(a_1,b) \bot f(a_3, b')$. But this contradicts to our assumption.
%\item Assume that there are $a_1 \bot a_2$, $a_1 \prec a_2$ and $b$ such that $f(a_1,b) < f(a_2,b)$ holds. Then there are elements $a_3$ and $b'$ with $a_2 < a_3$, $a_1 \bot a_3$ and $b' > b$. Then $f(a_2,b) < f(a_3,b')$ and $f(a_1,b) \bot f(a_3, b')$. But this contradicts to our assumption.

\end{enumerate}
\end{proof}

\begin{definition}
Let us say a binary function is \textnormal{\botfall}, if it has the same behaviour as $e_<$ respectively $e_\leq$ on pairs of partial type $(p_{\neq},p_{\neq})$.
\end{definition}

\begin{lemma} \label{lemma:botfall}
Let $f \in \Pol(\Gamma)$ be a canonical function $f: (P;\leq,\prec)^2 \to (P;\leq)$ of \botfall{} behaviour. Then $\Pol(\Gamma)$ contains $e_<$ or $e_\leq$.
\end{lemma}

\begin{proof}
From Lemma \ref{lemma:behaviour1} (7) follows that $f(p_\bot, p_=) = p_\bot$ and $f(p_=, p_\bot) = p_\bot$. By Lemma \ref{lemma:notdominated} we further know that $f(p_<,p_=), f(p_=,p_<) \in \{ p_{\bot}, p_{<} \}$. So we have to do a simple case distinction:
\begin{itemize}
\item If $f(p_=,p_<) = f(p_<,p_=) = p_{\bot}$, then $f$ behaves like $e_<$, hence $e_< \in \Pol(\Gamma)$. 
\item If $f(p_=,p_<) = p_{<}$ and $f(p_<,p_=) = p_{\bot}$, the function $(x,y)\to f(f(x,y),x)$ has the same behaviour as $e_<$, thus $e_< \in \Pol(\Gamma)$.
\item Symmetrically if $f(p_=,p_<) = p_{\bot}$ and $f(p_<,p_=) = p_{<}$, the function $(x,y)\to f(f(y,x),y)$ has the same behaviour as $e_<$, thus $e_< \in \Pol(\Gamma)$.
\item If $f(p_=,p_<) = p_{<} = f(p_<,p_=) = p_{<}$, then $f$ has the same behaviour as $e_{\leq}$, thus $e_{\leq} \in \Pol(\Gamma)$.
\end{itemize}
\end{proof}

We now give a simple criterium for the existence of a canonical \botfall{} function in $\Pol(\Gamma)$. This criterium will allows us to finish the proof of Theorem \ref{theorem:binarycanonical}.

\begin{lemma} \label{lemma:condition_e_1}
Assume that for every $k > 1$, every pair of tuples $\bar a, \bar b \in P^k$ and every indices $p,q \in [k]$ with $a_p < a_q$ and $\neg(b_p \leq b_q)$ there exists a binary function $g \in \Pol(\Gamma)$ such that $g(a_p,b_p) \bot g(a_q,b_q)$ and for all $i, j \in [k]$:
\begin{enumerate}
\item $a_i < a_j$ implies $g(a_i,b_i) < g(a_j,b_j)$ or $g(a_i,b_i) \bot g(a_j,b_j)$,
\item $a_i \bot a_j$ implies $g(a_i,b_i) \bot g(a_j,b_j)$.
\end{enumerate}
Then $\Pol(\Gamma)$ contains $e_<$ and $e_\leq$.
\end{lemma}

\begin{proof}
First we are going to show that for all $\bar a, \bar b \in P^k$ there is a binary function $f \in \Pol(\Gamma)$ that has \botfall{} on $(\bar a,\bar b)$. To be more precise we want to construct an $f \in \Pol(\Gamma)$ such that:
\begin{itemize}
\item $f(a_i,b_i) < f(a_j,b_j)$ if $a_i < a_j$ and $b_i < b_j$,
\item $f(a_i,b_i) \bot f(a_j,b_j)$ if $a_i < a_j$ and $\neg(b_p \leq b_q)$.
\item $f(a_i,b_i) \bot f(a_j,b_j)$ if $a_i \bot a_j$ and $b_i \neq b_j$
\end{itemize}

We are going to construct $f$ by a recursive argument.\\

Let $f^{(0)}(x,y)  = g^{(0)}(x, y) = x$ and $\bar a^{(0)} = f^{(0)}(\bar a, \bar b)$. If already $f^{(0)}$ has the desired properties we set $f(x,y) = f^{(0)}(x,y)$ and are done.

Otherwise, in the $(k+1)$-th recursion step, we are given a function $f^{(k)}(x,y)$ and a tuple $\bar a^{(k)} = f^{(k)} (\bar a, \bar b)$. Let us assume that there are indices $p,q$ with $a_p < a_q$, $\neg(b_p \leq b_q)$ and $a_p^{(k)} < a_q^{(k)}$. Then by our assumption there is a function $g^{(k+1)}(x,y) \in \Pol(\Gamma)$ such that $g^{(k+1)}(a_p^{(k)},b_p) \bot g^{(k)}(a_p^{(k)},b_p)$. We set $f^{(k+1)}(x,y) = g^{(k)}(f^{(k)}(x,y),y)$ and $\bar a^{(k)} = f^{(k)} (\bar a, \bar b)$.

Note that by the properties (1) and (2) of the function $g^{k}$ the only possible cases for $f^{k}$ being not \botfall{} is the case above. It is clear that the recursion ends after finitely many steps.

So on every finite subset $X \times Y$ of $P^2$ we find a \botfall{} function. By a compactness argument there exists a $h \in \Pol(\Gamma)$ that is \botfall{} on $P^2$. It remains to show that there is also a canonical \botfall{} function in $\Pol(\Gamma)$.

By Theorem \ref{theorem:ramsey2} we have that $h$ is canonical on arbitrarily large substructures of $P^2$. Let $(F_n)_{n \in \omega}$ be an increasing sequence of finite substructures such that its union is equal to $P$. Then for every $n \in \omega$ there are $\alpha_{1}^{(n)},\alpha_{2}^{(n)} \in \Aut(\Gamma)$ such that $f \circ (\alpha_{1}^{(n)},\alpha_{2}^{(n)})$ is canonical on $F_n$. By thinning out the sequence we can assume that $f \circ (\alpha_{1}^{(n)},\alpha_{2}^{(n)})$ has the same behaviour for every $n \in \omega$.

Since the behaviour $f \circ (\alpha_{1}^{(n)},\alpha_{2}^{(n)})$ on all $F_n$ is the same, we can inductively pick automorphisms $\beta_n \in \Aut(\Po)$ such that $\beta_n \circ f \circ (\alpha_{1}^{(n)},\alpha_{2}^{(n)})$ agrees with $\beta_{n+1} \circ f \circ (\alpha_{1}^{(n+1)},\alpha_{2}^{(n+1)})$ on $F_n$. The limit of this sequence is a canonical function in $\Pol(\Gamma)$ with \botfall{} behaviour.

By Lemma \ref{lemma:botfall} we have that $e_<$ or $e_\leq$ is an element of $\Pol(\Gamma)$. This concludes the proof.
\end{proof}

\begin{proof}[Proof of Theorem \ref{theorem:binarycanonical}]
Let $f: (P;\leq,\prec)^2 \to (P;\leq)$ be canonical and $f \in \Pol(\Gamma)$. Let us assume that $f$ is not dominated. By Lemma \ref{lemma:notdominated} we know $f(p_<,p_>) = f(p_<,p_{\bot, \succ}) = f(p_{\bot, \prec},p_>) = p_\bot$. %We are then going to show that either $e_< \in \Pol(\Gamma)$ or $e_{\leq} \in \Pol(\Gamma)$ with the help of Lemma \ref{lemma:condition_e_1}. \\ %respectively Lemma \ref{lemma:condition_e_2}.\\

By Lemma \ref{lemma:behaviour1} (3) and (4) we have to look at the following cases:
\begin{enumerate}
\item $f(p_<,p_{\bot, \prec}) = f(p_{\bot, \prec},p_<) = p_\bot$. 
\item $f(p_<,p_{\bot, \prec}) = p_<$ and $f(p_{\bot, \prec},p_<) = p_\bot$.
\item $f(p_<,p_{\bot, \prec}) = p_\bot$ and $f(p_{\bot, \prec},p_<) = p_<$.
\end{enumerate}

In the first case $f$ has \botfall{} behaviour therefore we are done by Lemma \ref{lemma:botfall}.

For the remaining cases we can restrict ourselves to (2), otherwise we take $(x,y) \to f(y,x)$. From Lemma \ref{lemma:behaviour1} (7) follows that $f(p_\bot, p_=) = p_\bot$. Thus $f(p_\bot, q) = p_\bot$ holds for every 2-type $q$. %By Lemma \ref{lemma:notdominated} we know that $f(p_<,p_=), f(p_=,p_<), f(p_=, p_{\bot, \prec}) \in \{ p_{\bot}, p_{<} \}$.

We are going to show that then the conditions in Lemma \ref{lemma:condition_e_1} are satisfied.  Let $\bar a, \bar b \in P^k$ be two tuples of arbitrary length $k$ and let $p,q \in [k]$ such that $a_p < a_q$, $b_p \prec b_q$ and $b_p \bot b_q$ hold. Then let $\alpha \in \Aut(\Po)$ with $\alpha(b_p) \succ \alpha(b_q)$. Such an automorphism exists by the homogeneity of $\Po$. Then we set $g(x,y) = f(x,\alpha(y))$. 

Clearly $g(a_p,a_q) \bot g(b_p,b_q)$, since $\alpha(b_p) \succ \alpha(b_q)$. Also the other conditions in Lemma \ref{lemma:condition_e_1} are satisfied, by the properties of $f$. Therefore $\Pol(\Gamma)$ contains $e_<$ or $e_\leq$.
\end{proof}

\section{The NP-hardness of $\Low$}

Let $\Low$ be the 3-ary relation defined by
\[ \Low(x,y,z) := (x < y \land z \bot xy) \lor (x < z \land y \bot xz). \]

Clearly $\bot$ and $<$ are pp-definable in $\Low$. Note that $\Low$ is not preserved by $e_<$ or $e_\leq$, so $\Csp(P;\Low)$ is not covered by the tractability result in Theorem \ref{theorem:horn2}. In this section we prove the NP-hardness of $\csp(P;\Low)$.

\begin{lemma} \label{lemma:abv}
Let us define the relations
\begin{align*}\Abv(x,y,z):=&(y<x\wedge xy\bot z)\vee(z<x\wedge xz\bot y)\\
U(x,y,z):=& (y<x \vee z < x)\wedge(y \bot z)
\end{align*}
Then $\Abv$ and $U$ are pp-definable in $\Low$.
\end{lemma}
\begin{proof}
Note that the formula
\[ \phi(x,y,z,v) := \exists u~u \bot v \land \Low(u,y,z) \land \Low(y,x,v) \land \Low(z,x,v) \]
is equivalent to the statement that $v\bot x$ and $y \bot z$ and at least one element of $\{y,z\}$ is smaller that $x$ and at most one element of $\{y,z\}$ is smaller than $v$\\
With that in mind one can see that
\[\exists v_1,v_2~\phi(x,y,z,v_1) \land \phi(v_2,y,z,x) \]
is equivalent to $\Abv(x,y,z)$ and 
\[ \exists v~\phi(x,y,z,v) \]
is equivalent to $U(x,y,z)$.
%First, we prove that $\Abv\in \langle (P;\Low)\rangle$. For a contradiction we assume otherwise that there is a polymorphism $f\in \Pol^{(2)}(P;\Low)$ that violates $\Abv$. By using the same argument as in the proof of Theorem \ref{theorem:symcanexistence}, $\{f\}\cup \Aut(\mP)$ generates $e_<$ or $e_{\leq}$, where $e_{<}$ and $e_{\leq}$ denote the embedding from $(P;<)^2$ to $(P;<)$ and the embedding from $(P;\leq)^2$ to $(P;\leq)$, respectively. Therefore $e_<$ or $e_{\leq}$ is a polymorphism of $(P;\Low)$, a contradiction. The case $\Low\in \langle(P;\Abv) \rangle$ can be argued similarly.
\end{proof}

%Let $I(x_1,x_2,\dots,x_k)$ denote the formula $\underset{1\leq i<j\leq k}{\wedge} x_i\bot x_j$.

\begin{theorem} \label{theorem:lowhard}
Let $a,b \in P$ with $a \bot b$. There is a pp-interpretation of $(\{0,1\};\rm 1IN3)$ in $(P;\Low,a,b)$. Thus $\csp(P;\Low)$ is NP-hard.
\end{theorem}

\begin{proof}
Let ${\rm NAE}$ be the Boolean relation $\{0,1\}^3\backslash \{(0,0,0),(1,1,1)\}$. It is easy to see that $\Pol(\{0,1\}, \rm NAE, 0, 1)$ is the projection clone $\prcl$. So by Theorem \ref{theorem:doubleshrink} it suffices to show that $(\{0,1\};{\rm NAE,0,1})$ has a pp-interpretation in $(P;\Low,a,b)$ to prove the Lemma.

Let $D:=\{x\in P:\Low(x,a,b)\},D_0:=\{x\in D:x<a\},D_1:=\{x\in D:x<b\}$. Note that $D_0\bot D_1$. Let $I:D\to \{0,1\}$ be given by:
\[
I(x):=
\begin{cases}
0&\text{if}~x\in D_0\\
1&\text{if}~x\in D_1
\end{cases}.
\]

Clearly the domain $D$ of $I$ is pp-definable in $(P;\Low,a,b)$. Since the order relation $<$ is pp-definable in $\Low$ also the sets $D_0$ and $D_1$ are pp-definable. Let $R = \{(x,y,z,t)\in P^4: (x>y\vee x>z\vee x>t) \land \neg(x \leq yzt) \}$. We claim that the relation $R$ is pp-definable in $\Low$. Observe that $(x,y,z,t) \in R$ is equivalent to
\[\exists u,v~(\Abv(x,u,v)\land U(x,y,u) \land U(x,z,u) \land U(x,t,v)) \]
and therefore pp-definable in $\Low$ by Lemma \ref{lemma:abv}. By the definition of $R$ we have that $I(c_1,c_2,c_3) \in \rm NAE$ if and only if $(a,c_1,c_2,c_3)\in R$ and $(b,c_1,c_2,c_3)\in R$. Thus the preimage of $\rm NAE$ is pp-definable in $(P;\Low,a,b)$.
\end{proof}

The following lemma gives us an additional characterization of reducts, in which $\Low$ is pp-definable.

\begin{lemma} \label{lemma:lowdominated}
The relation $\Low$ is pp-definable in $\Gamma$ if and only if every binary polymorphism of $\Gamma$ is dominated.
\end{lemma}

\begin{proof}
Every dominated function $f:P^2 \to P$ preserves $\Low$. 
For the other direction observe that by Lemma \ref{lemma:dominated} we have that $f$ is dominated in the first argument if and only if $f(a_1,b_1) < f(a_2,b_2)$ for all $a_1 < a_2$ and $b_1 \bot b_2$. Note that Lemma \ref{lemma:dominated} also works for non-canonical functions.\\
So if $f \in \Pol(\Gamma)$ is a binary, not dominated function, there are $a_1 < a_2, b_1 \bot b_2, a_1' \bot a_2' $ and $b_1' < b_2'$ such that $f(a_1,b_1) \bot f(a_2,b_2)$ and $f(a_1',b_1') \bot f(a_2',b_2')$.
Hence $f$ violates the relation 
\[S(x_1,x_2,y_1,y_2) := (x_1 < x_2 \land y_1 \bot y_2) \lor (x_1 \bot x_2 \land y_1 < y_2). \]
But the relation $S$ and $\Low$ are pp-interdefinable:
\begin{align*}
\Low(x,y,z) \leftrightarrow & S(x,y,x,z) \land y \bot z \\
S(x_1,x_2,y_1,y_2) \leftrightarrow & \exists u,v,w~(\Low(x_1,x_2,u) \land \Abv(u,x_1,v), \\
&\land \Low(u,v,w) \land \Abv(w,y_1,v) \land \Low(y_1,y_2,w)).
\end{align*}
We conclude that $f$ violates $\Low$.
\end{proof}

\section{Violating the $\Low$ relation}
We saw in Theorem \ref{theorem:horn2} that $\csp(\Gamma)$ is tractable if $e_<$ or $e_\leq$ are polymorphisms of $\Gamma$. By Theorem \ref{theorem:lowhard} we know that $\csp(\Gamma)$ is NP-complete if $\Low$ is pp-definable in $\Gamma$. In this section we are going to show that these results already cover all possible reducts where $<$ and $\bot$ are pp-definable.

\begin{theorem}\label{theorem:symcanexistence}
Let $\Gamma$ be a reduct of $\Po$ such that $\bot$ and $<$ are pp-definable in $\Gamma$.  Then $\Low$ is not pp-definable in $\Gamma$ if and only if $\pol(\Gamma)$ contains one of the functions $e_<$ or $e_{\leq}$.
\end{theorem}

\begin{proof}
Note that by Theorem \ref{theorem:violate} $\Low$ is not pp-definable in $\Gamma$ if and only if there is a binary $f \in \Pol(\Gamma)$ violating $\Low$. This means that there are $a,b,c\in P$ such that $a<b\wedge ab\bot c$ and $f(a,a)<f(b,c)\wedge f(a,a)<f(c,b)$, or $f(a,a)\bot f(b,c)$ and $f(a,a)\bot f(c,b)$. 

We have only these two cases since $f$ preserves $\lowbot$ and $\bot$. We can assume that $a\prec b\prec c$ since otherwise we can find an automorphism $\alpha\in \Aut(\mP)$ such that $\alpha(a)\prec \alpha(b)\prec \alpha(c)$. Then we consider the map $(x,y)\mapsto f(\alpha^{-1}(x),\alpha^{-1}(y))$ with three elements $\alpha(a)$, $\alpha(b)$ and $\alpha(c)$ instead. 

By Theorem \ref{theorem:ramsey} we can assume that $f$ is canonical as a function from $(P;<,\prec,a,b,c)^2$ to $(P;<)$. We deal with the two cases in Lemma \ref{lem:impos1} and Lemma \ref{lem:impos2} in the following subsections.
\end{proof}

\begin{notation}
For simplicities sake, a canonical binary function in this section means a function that is canonical as a function from $(P;\leq,\prec)^2 \to (P;<)$.\\
Let $f:P^2\to P$ be a function and $X,Y,X',Y'$ be subsets of $P$ such that $f$ is dominated on $X\times Y$ and $X'\times Y'$. We say that $f$ has the \textit{same domination} on $X\times Y$ and $X'\times Y'$ if $f$ is dominated by the first argument on both $X\times Y$ and $X'\times Y'$, or dominated by the second argument on both $X\times Y$ and $X'\times Y'$. Otherwise, we say that $f$ has the \textit{different domination} on $X\times Y$ and $X'\times Y'$.
\end{notation}

%\begin{lemma} \label{lemma:condition_e_2}
%Assume that for every arity $k$, every pair of tuples $\bar a, \bar b \in P^k$ and every indices $p,q \in [k]$ with $a_p < a_q$ and $\neg(b_p \leq b_q)$ there exists a binary $g \in \Pol(\Gamma)$ such that $g(a_p,b_p) \bot g(a_q,b_q)$ holds and for all $i,j \in [k]$:
%\begin{enumerate}
%\item $a_i \leq a_j$ and $b_i \leq b_j$ implies $g(a_i,b_i) \leq g(a_j,b_j)$
%\item $a_i < a_j$ implies $g(a_i,b_i) < g(a_j,b_j)$ or $g(a_i,b_i) \bot g(a_j,b_j)$,
%%\item $a_i = a_j$ and $b_i < b_j$ implies $g(a_i,b_i) < g(a_j,b_j)$,
%\item $a_i = a_j$ and $b_i \bot b_j$ implies $g(a_i,b_i) \neq g(a_j,b_j)$,
%\item $a_i \bot a_j$ implies $g(a_i,b_i) \bot g(a_j,b_j)$.
%\end{enumerate}
%And suppose that also for every $a_p = a_q$ with $b_p < b_q$ a binary $g \in \Pol(\Gamma)$ exists such that $g(a_p,b_p) < g(a_q,b_q)$ and (1)-(4) holds. In this case the embedding $e_{\leq}(x,y)$ is an element of $\Pol(\Gamma)$. 
%\end{lemma}
%
%\begin{proof}
%The proof uses a recursive argument very similar to the proof of Lemma \ref{lemma:condition_e_1} and is left to the reader. 
%\end{proof}

\subsection{$f(a,a)<f(b,c)\wedge f(a,a)<f(c,b)$}  \label{sec:firstcase} \
The aim of this subsection is to prove the following lemma.
\begin{lemma}\label{lem:impos1}
% If every binary canonical function in $\pol (\Gamma)$ is dominated, then $\neg(f(a,a)<f(b,c)\wedge f(a,a)<f(c,b))$.
Let $f \in \Pol(\Gamma)$ be canonical as a function from $(P;<,\prec,a,b,c)^2$ to $(P;<)$. If $f(a,a)<f(b,c)\wedge f(a,a)<f(c,b)$ then $\pol(\Gamma)$ contains $e_<$ or $e_\leq$.\end{lemma}

We define the following two sets: 
\begin{itemize}
% \item $B_1:=\{x\in P:x>b\wedge x\prec c\wedge x\bot c\}$.
% \item $B_2:=\{x\in P:x>b\wedge x>c\}$.
% \item $B:=B_1\cup B_2$.
% \item $C:=\{x\in P:x>c\wedge x\bot b\wedge x\bot a\}$.
\item $B_1:=\{x\in P:x>c\wedge x\bot a\wedge x\bot b\}$,
\item $B_2:=\{x\in P:x>b\wedge x>c\}$.
%\item $B_3:=\{x\in P:x>b\wedge x\bot c\wedge x\prec c\}$. We don't need B_3
\end{itemize}
%Note that $(B;<,\prec)$ and $(C;<,\prec)$ are isomorphic to $(P;<,\prec)$.
Let $x,y\in B_1\cup B_2$. We say that $x$ and $y$ \emph{are in the same orbit} if $x\in B_i$ and $y\in B_i$ for an $i\in [2]$. 
\begin{observation}\label{obs:firstobservation} \
$B_1$ and $B_2$ are orbits of $\Aut(P;<,\prec,a,b,c)$. By the homogeneity of $(P;\leq,\prec)$ we can show that $(B_1;\leq,\prec)$, $(B_2;\leq,\prec)$ are isomorphic to $(P;\leq,\prec)$. Further also the union of $B_1$ and $B_2$ is an isomorphic copy of $(P;\leq,\prec)$, in which $B_1$ forms a random filter.
\end{observation}
If there is a canonical $g \in \Pol(\Gamma)$ that is not dominated, then Lemma \ref{theorem:binarycanonical} gives us that $e_<$ or $e_\leq$ is in $\Pol(\Gamma)$. So throughout the lemmata and corollaries below in this section, we assume that every binary canonical function in $\Pol(\Gamma)$ is dominated and $f(a,a)<f(b,c)\wedge f(a,a)<f(c,b)$.

% \begin{lemma}\label{lem:exclusion}
% $\neg (f(a,a)<f(b,c)\wedge f(a,a)<f(c,b))$.
% \end{lemma}
% % In this section we study the property of $\pol(\Gamma)$ when $f$ satisfies $f(a,a)<f(b,c)\wedge f(a,a)<f(c,b)$. The aim of this subsection is to prove the following.
% For a contradiction, in the rest of this subsection, we assume that $f(a,a)<f(b,c)\wedge f(a,a)<f(c,b)$ hold.
\begin{lemma}\label{lem:firstdomination}
$f$ is dominated on $B_i\times B_j$ for every $i,j\in [2]$.
\end{lemma}
\begin{proof}
For a contradiction, we assume that $f$ is not dominated on $B_i\times B_j$. Since $(B_i;\leq,\prec)$ and $(B_j;\leq,\prec)$ are isomorphic to $(P;\leq,\prec)$, there are $\alpha:P\to B_i$ and $\beta:P\to B_j$ such that $\alpha$ is an isomorphism from $(P;\leq,\prec)$ to $(B_i;\leq,\prec)$ and $\beta$ is an isomorphism from $(P;\leq,\prec)$ to $(B_j;\leq,\prec)$. Let $g:P^2\to P$ be given by $g(x,y):=f(\alpha(x),\beta(y))$. It follows from Observation \ref{obs:firstobservation} that $g$ is canonical and is not dominated, a contradiction.
\end{proof}
\begin{lemma}\label{lem:domination}
$f$ has the same domination on all sets $B_i\times B_j$, $i,j\in [2]$. 
\end{lemma}
\begin{proof}
We claim that $f$ has the same domination on $B_1\times B_k$ and $B_{2}\times B_k$ for any $k\in [2]$. For a contradiction, we assume that $f$ does not have the same domination $B_1\times B_k$ and $B_{2}\times B_k$. Without loss of generality we can assume that $f$ is dominated by the first argument on $B_1\times B_k$ and dominated by the second argument on $B_{2}\times B_k$. Let $x,y\in B_1,z,t\in B_2$ be such that $x<y\wedge y<z\wedge x\bot t$. Let $x',y',z',t'\in B_k$ be such that $x'\bot t'\wedge y'<z'\wedge z'<t'$. Since $f$ is dominated by the first argument on $B_1\times B_k$, we have $f(x,x')<f(y,y')$. Since $f$ is dominated by the second argument on $B_2\times B_k$, we have $f(z,z')<f(t,t')$. Since $f$ preserves $<$, we have $f(y,y')<f(z,z')$. Thus $f(x,x')<f(t,t')$, a contradiction to the fact that $f$ preserves $\bot$.

By considering the map $(x,y)\mapsto f(y,x)$ we have that $f$ has the same domination on $B_k\times B_1$ and $B_k\times B_2$ for every $k\in [2]$. This implies that $f$ has the same dominations on all products $B_i\times B_j,i,j\in [2]$. 
\end{proof}

In the rest of this section, we assume that $f$ is dominated by the first argument on $B_i\times B_j$ for every $i,j\in [2]$. The other case can be reduced to this case by considering the map $(x,y)\mapsto f(y,x)$.
\begin{lemma}\label{lem:firstfiner}

Let $u,v\in B_1$ and $u'\in B_2 ,v'\in B_1$ be such that $u<v\vee u\bot v$. Then $f(u,u')\bot f(v,v')$.
\end{lemma}
\begin{proof}
First, we claim that $f(u,u')>f(v,v')\vee f(u,u')\bot f(v,v')$. For a contradiction, we assume that $f(u,u')\leq f(v,v')$. Since $f$ preserves $<$, we have $f(c,b)<f(u,u')$. Therefore $f(a,a)<f(c,b)<f(u,u')<f(v,v')$, a contradiction to the $\bot$-preservation of $f$. Thus the claim follows.

The proof is completed by showing that $f(u,u')>f(v,v')$ is impossible. For a contradiction, we assume that $f(u,u')>f(v,v')$. Let $s,t\in B_1$ be such that $s\bot t\wedge s<v\wedge u<t$. Let $s'\in B_1,t'\in B_2$ be such that $s'\bot t'$. By the domination of $f$, we have $f(s,s')<f(v,v')\wedge f(u,u')<f(t,t')$. It follows from $f(u,u')>f(v,v')$, we have $f(s,s')<f(t,t')$, a contradiction to $\bot$-preservation of $f$.
\end{proof}

\begin{lemma}\label{lem:secondfiner}
Let $u,v\in B_1$ be such that $u\bot v$. Then for every $u',v'\in B_1\cup B_2$, we have $f(u,u')\bot f(v,v')$.
\end{lemma}
\begin{proof}
For a contradiction, we assume that $\neg(f(u,u')\bot f(v,v'))$. Without loss of generality, we assume that $f(u,u')\leq f(v,v')$. Let $s,t\in B_1$ be such that $s<u\wedge v<t\wedge s\bot t$. Let $s',t'\in B_1\cup B_2$ be such that $s'\bot t'$, $s',u'$ are in the same orbit and $t',v'$ are in the same orbit. By the domination of $f$, we have $f(s,s')<f(u,u')\wedge f(v,v')<f(t,t')$. Since $f(u,u')<f(v,v')$, we have $f(s,s')<f(t,t')$, a contradiction to the $\bot$-preservation of $f$.
\end{proof}
\begin{lemma}\label{lem:thirdfiner}
Let $u,v\in B_1$ and $u',v'\in B_1\cup B_2$ be such that $u<v$. Then $f(u,u')<f(v,v')\vee f(u,u')\bot f(v,v')$. 
\end{lemma}
\begin{proof}
%True since f preserves <\bot
For a contradiction, we assume that $f(v,v')\leq f(u,u')$. Let $s,t\in B_1$ be such that $t<v\wedge u<s\wedge s\bot t$. Let $s',t'\in B_1\cup B_2$ be such that $s'\bot t'$, $s',u'$ are in the same orbit, and $t',v'$ are in the same orbit. By the domination of $f$, we have $f(t,t')<f(v,v')\wedge f(u,u')<f(s,s')$. Since $f(v,v')<f(u,u')$, we have $f(t,t')<f(s,s')$, a contradiction to the $\bot$-preservation of $f$.
\end{proof}

%\begin{lemma}\label{lem:forthfiner}
%Let $u\in B_1$ and $v'\in B_1\land u' \in B_2$. Then $f(u,u')\bot f(u,v') \lor f(u,v') < f(u,u')$. 
%\end{lemma}
%\begin{proof}
%We first claim that $f(u,u')>f(u,v')\vee f(u,u')\bot f(u,v')$. For a contradiction, we assume that $f(u,u')\leq f(u,v')$. Since $f$ preserves $<$, we have $f(c,b)<f(u,u')$. Therefore $f(a,a)<f(c,b)<f(u,u') \leq f(u,v')$, a contradiction to the $\bot$-preservation of $f$. Thus the claim follows.
%\end{proof}

\begin{proof}[Proof of Lemma \ref{lem:impos1}]
We are going to show that $\Pol(\Gamma)$ contains a function that behaves like $e_<$ or like $e_\leq$ by checking the conditions of Lemma \ref{lemma:condition_e_1}.

So let $\bar a,\bar b \in P^k$ with $a_p < a_q$ and $\neg(b_p \leq b_q)$. We set $Y:=\{b_i:b_i\geq b_p\}, Z:=\{b_i:\neg(b_i\geq b_p)\}$. By definition we have $b_q\in Z$. By the homogeneity of $\mP$, there is $\alpha\in \Aut(\mP)$ such that $\alpha(Y)\subseteq B_2$ and $\alpha(Z)\subseteq B_1$. Let $\beta\in \Aut(\mP)$ such that $\beta(\{a_i:i\in [k]\})\subseteq B_1$. Let $g(x,y):=f(\beta(x),\alpha(y))$. Clearly, $g\in \pol(\Gamma)$. \\
By Lemma \ref{lem:firstfiner} we have that $g(a_p,b_p) \bot g(a_q,b_q)$. Further we know by Lemma \ref{lem:thirdfiner} that $g(a_i,b_i) < g(a_j,b_j)$ or $g(a_i,b_i) \bot g(a_j,b_j)$ holds for all $a_i < a_j$. By Lemma \ref{lem:secondfiner} we know that $g(a_i,b_i) \bot g(a_j,b_j)$ holds for all $a_i \bot a_j$.  So the conditions of Lemma \ref{lemma:condition_e_1} are satisfied. Hence $e_<$ or $e_\leq$ is a polymorphism of $\Gamma$.

\end{proof}

\subsection{$f(a,a)\bot f(b,c)\wedge f(a,a)\bot f(c,b)$}
\label{sec:secondcase}
The aim of this section is to prove the following.
\begin{lemma}\label{lem:impos2}
% If every binary canonical function in $\pol(\Gamma)$ is dominated, then $\neg (f(a,a)\bot f(b,c)\wedge f(a,a)\bot f(c,b))$.
Let $f \in \Pol(\Gamma)$ be canonical as a function from $(P;<,\prec,a,b,c)^2$ to $(P;<)$. If $f(a,a)\bot f(b,c)\wedge f(a,a)\bot f(c,b)$, then $\pol^{(2)}(\Gamma)$  contains $e_<$ or $e_\leq$.
\end{lemma}
% For a contradiction, throughout this subsection, {\bf we assume that every binary canonical function in $\pol(\Gamma)$ is dominated and $f(a,a)\bot f(b,c)f(c,b)$.}
 We define the following sets. 
\begin{align*}
B_1&:=\{x\in P:a<x<b\wedge x\bot c\}\\
%B_2&:=\{x\in P:a<x<b\wedge x<c\}\\ this set is not needed (and empty, since a \bot c)
B_2&:=\{x\in P: x < b \wedge x < c \wedge x \bot a \wedge x \prec a\}.
\end{align*}
%Let $x,y\in B_1\cup B_2$. We say that $x,y$ \emph{are in the same region} if $x\in B_i\Leftrightarrow y\in B_i$ for every $i\in [2]$.

Throughout the lemmata and corollaries below in this section, we assume that every binary canonical function in $\Gamma$ is dominated and $f(a,a) \bot f(b,c)\wedge f(a,a) \bot f(c,b)$.\\

Observe that by the homogeneity of $(P;\leq; \prec)$ and the back-and-forth argument, we can show that $(B_1\cup B_2;\leq,\prec)$ is isomorphic to $(P;\leq,\prec)$, with $B_2$ being a random filter.
For every two $k$-tuples $\bar x$ and $\bar y$ in $B_i^k$, $\bar x$ and $\bar y$ are in the same orbit of $\Aut(\mP)$ if and only if $\bar x$ and $\bar y$ are in the same orbit of $\Aut(P;a,b,c)$.

\begin{lemma}
$f$ has the same domination on sets $B_i\times B_j,i,j\in [2]$.
\end{lemma}
\begin{proof}
 This lemma can be shown as in Lemma \ref{lem:firstdomination} and Lemma \ref{lem:domination}.
\end{proof}
% \begin{proof}[Proof of Lemma \ref{lem:impos2}]
% Without loss of generality, we can assume that $f$ is dominated by the first argument on $B_i\times B_j$ for every $i,j\in \{1,2,3\}$. Let $u,v,u',v'\in B_1$ be such that $u<v\wedge u'<v'$. Since $f$ is dominated by the first argument on $B_1\times B_1$, we have that $f(u,u')<f(v,v')$. Since $f$ preserves $<$ and $a<u\wedge a<u'$, we have that $f(a,a)<f(u,u')$. Since $$
% \end{proof}
In the rest of this section we assume that $f$ is dominated by the first argument on $B_i\times B_j$ for every $i,j\in \{1,2\}$. Similarly, to Lemma \ref{lem:firstfiner}, we have the following.

\begin{lemma}\label{lem:dirreversion}
Let $u,v\in B_1$ and $u'\in B_1,v'\in B_2$ be such that $u<v\vee u\bot v$. Then $f(u,u')\bot f(v,v')$.
\end{lemma}
\begin{proof}
First we prove that $f(v,v')<f(u,u') \vee f(v,v')\bot f(u,u')$. For a contradiction we assume that $f(u,u')\leq f(v,v')$. Since $a<u\wedge a<u'$, we have $f(a,a)<f(u,u')$. Since $v<b\wedge v'<c$, we have $f(v,v')<f(b,c)$. Thus $f(a,a)<f(b,c)$, a contradiction to the fact that $f(a,a)\bot f(b,c)$. Thus $f(v,v')<f(u,u')\vee f(v,v')\bot f(u,u')$.

The proof is completed by showing that $f(u,u')>f(v,v')$ is impossible. For a contradiction, we assume that $f(u,u')>f(v,v')$. Let $s,t\in B_1$ be such that $s\bot t\wedge s<v\wedge u<t$. Let $s'\in B_2,t'\in B_1$ be such that $s'\bot t'$. By the domination of $f$, we have $f(s,s')<f(v,v')\wedge f(u,u')<f(t,t')$. It follows from $f(u,u')>f(v,v')$, we have $f(s,s')<f(t,t')$, a contradiction to $\bot$-preservation of $f$.
\end{proof}

\begin{lemma}\label{lem:botpreservation}
Let $u,v\in B_1$ be such that $u\bot v$. Then for every $u',v'\in B_1\cup B_2$, we have $f(u,u')\bot f(v,v')$.
\end{lemma}
\begin{proof}
analogous to Lemma \ref{lem:secondfiner}.
%For a contradiction, we assume that $\neg(f(u,u')\bot f(v,v'))$. Without loss of generality, we assume that $f(u,u')\leq f(v,v')$. Let $s,t\in B_1$ be such that $s<u\wedge v<t\wedge s\bot t$. Let $s',t'\in B_1\cup B_2\cup B_3$ be such that $s'$ and $u'$ are in the same region, $t'$ and $v'$ are in the same region, and $s'\bot t'$. By the domination of $f$, we have $f(s,s')<f(u,u')\wedge f(v,v')<f(t,t')$. It follows from the assumption $f(u,u')\leq f(v,v')$, we have $f(s,s')<f(t,t')$, a contradiction to the fact that $f$ preserves $\bot$. 
\end{proof}
\begin{lemma}\label{lem:dircompatible}
Let $u,v\in B_1$ and $u',v'\in B_1\cup B_2$ be such that $u<v$. Then $f(u,u')<f(v,v')\vee f(u,u')\bot f(v,v')$.
\end{lemma}
\begin{proof}
analogous to Lemma \ref{lem:thirdfiner}.
%For a contradiction, we assume that $f(v,v')\leq f(u,u')$. Let $s,t\in B_1$ be such that $s<v\wedge u<t$. Let $s',t'\in P$ be such that $s'$ and $v'$ are in the same region, $t'$ and $u'$ are in the same region, and $s'\bot t'$. By the domination of $f$, we have $f(s,s')<f(v,v')\wedge f(u,u')<f(t,t')$. Since $f(v,v')\leq f(u,u')$, it follows that $f(s,s')<f(t,t')$. It contradicts the fact that $f$ preserves $\bot$.
\end{proof}

%\begin{lemma}\label{lem:direqual}
%Let $u\in B_1$ and $u'\in B_1\land v' \in B_2$. Then $f(u,u')\bot f(u,v') \lor f(u,u') > f(u,v')$. 
%\end{lemma}
%\begin{proof}
%First we claim that $f(u,u') > f(u,v') \vee f(u,u')\bot f(u,v')$. For a contradiction we assume that $f(u,u')\leq f(u,v')$. Since $a<u\wedge a<u'$, we have $f(a,a)<f(u,u')$. Since $u<b\wedge v'<c$, we have $f(u,v')<f(b,c)$. Thus $f(a,a)<f(b,c)$, a contradiction to the fact that $f(a,a)\bot f(b,c)$. Thus $f(u,v')<f(u,u')\vee f(v,v')\bot f(u,u')$.
%\end{proof}

\begin{proof}[Proof of Lemma \ref{lem:impos2}]
We are again going to show that $\Pol(\Gamma)$ contains a function that behaves like $e_<$ or like $e_\leq$ by checking the conditions of Lemma \ref{lemma:condition_e_1}.

So let $\bar a,\bar b \in P^k$ with $a_p < a_q$ and $\neg(b_p \leq b_q)$. We set $Y:=\{b_i:b_i\geq b_p\}, Z:=\{b_i:\neg(b_i\geq b_p)\}$. By definition we have $b_q\in Z$. By the homogeneity of $\mP$, there is $\alpha\in \Aut(\mP)$ such that $\alpha(Y)\subseteq B_1$ and $\alpha(Z)\subseteq B_2$. Let $\beta\in \Aut(\mP)$ such that $\beta(\{a_i:i\in [k]\})\subseteq B_1$. Let $g(x,y):=f(\beta(x),\alpha(y))$. Clearly, $g\in \pol(\Gamma)$. \\
By Lemma \ref{lem:firstfiner} we have that $g(a_p,b_p) \bot g(a_q,b_q)$. Further we know by Lemma \ref{lem:thirdfiner} that $g(a_i,b_i) < g(a_j,b_j)$ or $g(a_i,b_i) \bot g(a_j,b_j)$ holds for all $a_i < a_j$. By Lemma \ref{lem:secondfiner} we know that $g(a_i,b_i) \bot g(a_j,b_j)$ holds for all $a_i \bot a_j$.  So the conditions of Lemma \ref{lemma:condition_e_1} are satisfied. Hence $e_<$ or $e_\leq$ is a polymorphism of $\Gamma$.
\end{proof}

\section{The NP-hardness of $\Betw$, $\sept$ and $\cycl$}\label{sect:cycl}

By Corollary \ref{corollary:monoids} we are now left with the cases where $\End(\Gamma)$ is equal to one of the monoids $\overline{\langle \updownarrow \rangle}$, $\overline{\langle \circlearrowright \rangle}$ or $\overline{\langle \updownarrow, \circlearrowright \rangle}$.

 We are going to deal with all these remaining cases in this section. Interestingly, we can treat them all similarly: By fixing finitely many constants $c_1,\ldots,c_n$ on $\Gamma$ we obtain definable subsets of $(\Gamma,c_1,\ldots,c_n)$ on which $<$ and $\Low$ are pp-definable. This enables us to reduce every every such case to the NP-completeness of $\Low$.

\begin{lemma} \label{lemma:betwhard} Let $u,v \in P$ with $u < v$. Then the relations $<$ and $\Low$ are pp-definable in $(P,\Betw,\bot,u,v)$.
\end{lemma} 

\begin{proof}
It is easy to verify that there is a pp-definition of the order relation by the following equivalence:
\[ x < y \leftrightarrow \exists a,b~(\Betw(x,y,a) \land \Betw(y,a,b) \land \Betw(u,v,a) \land \Betw(v,a,b)). \]

The two maps $e_<:P^2 \to P$ and $e_\leq: P^2 \to P$ do not preserve $\Betw$, since for every triple $ \bar a = (a_1,a_2,a_3)$ with $a_1 < a_2 < a_3$ and $\bar b = (b_1,b_2,b_3)$ with $b_1 > b_2 >b_3$, the image of $(\bar a,\bar b)$ forms an antichain.

By Theorem \ref{theorem:symcanexistence} we have that $\Low$ is pp-definable in $(P,\Betw,\bot,u,v)$.
\end{proof}

\begin{theorem}  \label{theorem:betwhard}
Let $\Gamma$ be a reduct of $\Po$ such that $\End(\Gamma) = \overline{\langle \updownarrow \rangle}$. Then there are constants $u,v,w,t \in P$ such that $(\{0,1\},\rm 1IN3)$ is pp-interpretable in $(\Gamma, u,v,w,t)$. Hence $\csp(\Gamma)$ is NP-complete.
\end{theorem}

\begin{proof}
Note that the betweenness relation $\Betw$ is an orbit of $\End(\Gamma) = \overline{\langle \updownarrow \rangle}$ on $P^3$. Now Theorem \ref{theorem:violate} implies that $\Betw$ is primitively positive definable in $\Gamma$. For the same reason $\bot$ is pp-definable in $\Gamma$. By Lemma \ref{lemma:betwhard} there is pp-definition of $\Low$ in $(\Gamma,u,v)$. By Theorem \ref{theorem:lowhard} we can find a pp-interpretation of $(\{0,1\},\rm 1IN3)$ in $(\Gamma,u,v,w,t)$, where $w,t$ are two additional constants. Hence $\csp(\Gamma)$ is NP-complete.
\end{proof}

For the case where $\End(\Gamma) = \overline{\langle \circlearrowright \rangle}$, we first need the following lemma:

\begin{lemma}\label{lemma:cyclhard}
Let $c,d$ be two constants in $P$ such that $c<d$. Then there is a pp-interpretation of $(P;\Low)$ in $(P;\cycl,c,d)$
%The relation $(\{0,1\};{\rm 1IN3})$ is pp interpretable in $(P;<,\cycl,a,b)$, where $a,b$ are two constants in $P$ such that $a\bot b$.
\end{lemma}

\begin{proof}
Let $X:=\{x\in P:c<x<d\}$. By using back-and-forth argument one can show easily that $(P;<)$ and $(X;<_{|X})$  are isomorphic. We first show that $X$ (as a unary predicate) and $<_{|X}$ are pp-definable in $(P;\cycl,c,d)$. It is easy to verify that the set $X$ can be defined in $(P;\cycl,c,d)$ by $\phi(x):=\cycl(c,x,d)$ and that $x<_{|X}y \leftrightarrow \phi(x)\wedge \phi(y)\wedge \cycl(c,x,y)$. Now a pp-interpretation of $(P;<,\cycl)$ in $(P;\cycl,c,d)$ is simply given by the identity on $X$.

By Lemma \ref{lemma:ppdefpar} we have that $\bot$ is pp-definable in $(P;<,\cycl)$. It is easy to verify that $e_<$ and $e_\leq$ do not preserve $\Cycl$. Therefore, by Theorem \ref{theorem:symcanexistence}, $\Low$ is pp-definable in $(P;<,\cycl)$, which concludes the proof of the Lemma.
%By Theorem \ref{theo:posetaut} $\Aut(P;<,\cycl)$ must be one of the five groups in Theorem \ref{theo:posetaut}. The automorphism group $\Aut(P;<,\cycl)$ cannot be one of the four last groups since otherwise $<$ is violated. Thus $\Aut(P;<,\cycl)$ must be $\Aut(\mP)$. In this case $\bot$ is in one orbit of $2$-tuples of $\Aut(\mP)$. It follows from Theorem \ref{theorem:violate} that $\bot$ is pp definable in $\Aut(P;<,\cycl)$. By Theorem \ref{theorem:symcanexistence} the relation $\Low$ is pp definable in $(P;<,\cycl)$. It follows from Theorem \ref{theorem:lowhard} that $(\{0,1\};{\rm 1IN3})$ has a pp interpretation in $(P;\Low,a,b)$, therefore has a pp interpretation in $(P;<,\cycl,a,b)$.
\end{proof}

\begin{theorem} \label{theorem:cyclhard}
Let $\Gamma$ be a reduct of $\Po$ such that $\End(\Gamma) = \overline{\langle \circlearrowright \rangle}$. Then there are constants $a,b,c,d \in P$ such that $(\{0,1\},\rm 1IN3)$ is pp-interpretable in $(\Gamma, a,b,c,d)$. Hence $\csp(\Gamma)$ is NP-complete.
\end{theorem}

\begin{proof}
The cyclic order relation $\cycl$ is an orbit of $\End(\Gamma) = \overline{\langle \circlearrowright \rangle}$ on $P^3$. So Theorem \ref{theorem:violate} implies that $\Cycl$ is primitively positive definable in $\Gamma$. By Lemma \ref{lemma:cyclhard} there is pp-definition of $\Low$ in $(\Gamma,c,d)$ with $c<d$. By Theorem \ref{theorem:lowhard} we can find a pp-interpretation of $(\{0,1\},\rm 1IN3)$ in $(\Gamma,a,b,c,d)$, where $a,b$ are two additional constants. Hence $\csp(\Gamma)$ is NP-complete.
\end{proof}

In the following, we prove the NP-hardness of $\csp(P;\sept)$ by using the same proof idea as the proof of NP-hardness of $\csp(P;\cycl)$ in Section \ref{sect:cycl}.

\begin{lemma} \label{lemma:septhard}
Let $c,d,u$ be constants in $P$ such that $c<d<u$. Then $(P;\Low)$ has a pp-interpretation in $(P;\Sept,c,d,u)$.
\end{lemma}

\begin{proof}
Let $X:=\{x\in P:d<x<u\}$. By using a back-and-forth argument, one can show easily that $(X;\leq)$ and $\mP$ are isomorphic. Similarly as in the proof of Theorem \ref{theorem:cyclhard}, $X$ and $<_{|X}$ are pp definable in $(P;\sept,c,d,u)$ as follows.

The set $X$ can be defined by the formula $\phi(x):=\sept(c,d,x,u)$, and $<_{|X}$ can be defined by $x<_{|X}y:\Leftrightarrow \phi(x)\wedge \phi(y)\wedge \sept(c,d,x,y)$. Also $\Cycl(x,y,z)_{|X}$ can be defined by the primitive positive formula $\phi(x)\wedge \phi(y) \wedge \phi(z) \land \sept(c,x,y,z)$

So a pp-interpretation of $(P;<,\Cycl)$ in $(P;\Sept,c,d,u)$ is simply given by the identity, restricted to $X$. By Lemma \ref{lemma:cyclhard}, $\Low$ is pp-definable in $(P;<,\Cycl)$, which concludes the proof of the Lemma.
\end{proof}

\begin{theorem} \label{theorem:septhard}
Let $\Gamma$ be a reduct of $\Po$ such that $\End(\Gamma) = \overline{\langle \updownarrow, \circlearrowright \rangle}$. Then there are constants $a,b,c,d,u \in P$ such that $(\{0,1\},\rm 1IN3)$ is pp-interpretable in $(\Gamma, a,b,c,d,u)$. Hence $\csp(\Gamma)$ is NP-complete.
\end{theorem}

\begin{proof}
The relation $\Sept$ is an orbit of $\End(\Gamma) = \overline{\langle \updownarrow, \circlearrowright \rangle}$ on $P^3$. So Theorem \ref{theorem:violate} implies that $\Sept$ is primitively positive definable in $\Gamma$. By Lemma \ref{lemma:septhard} there is pp-definition of $\Low$ in $(\Gamma,c,d,u)$ with  $c<d<u$. By Theorem \ref{theorem:lowhard} we can find a pp-interpretation of $(\{0,1\},\rm 1IN3)$ in $(\Gamma,a,b,c,d,u)$, where $a,b$ are two additional constants. Hence $\csp(\Gamma)$ is NP-complete.
\end{proof}

\section{Main Results} \label{sect:mainresult}

In this section we complete the proof of the complexity dichotomy for the Poset-SAT($\Phi$) problems that we announced in Theorem \ref{theorem:main} and that we reformulated as CSPs on the reducts of the random poset $\Po$. We have proven an even stronger dichotomy that remains interesting even if P=NP. This dichotomy regards model-theoretic properties of the reducts of $\Po$ and can be also stated in terms of universal-algebra by what we saw in Section \ref{sect:topclones}. We will phrase it in Theorem \ref{theorem:algdic}.

\subsection{An algebraic dichotomy} \label{sect:algdic}

Let $\Gamma$ be a reduct of $\Po$ and $\Delta$ be its model-complete core. Throughout this paper we have studied the question whether there is a pp-interpretation of the structure $(\{0,1\};\rm 1IN3)$ in $\Delta$, extended by finitely many constants or not.\\
By Theorem \ref{theorem:topbirk} and Theorem \ref{theorem:doubleshrink} we know that this fact can be elegantly described with the help of topological clones. We sum up our results and show that - for reducts of the random poset - we can also give an additional characterization by weak near unanimity polymorphisms (modulo endomorphism). First we are going to look in detail at the case, where $<$ and $\bot$ are pp-definable.

\begin{lemma} \label{lemma:algdic}
Let $\Gamma$ be a reduct of $\Po$ in which $<$ and $\bot$ are pp-definable. Then the following are equivalent:
\begin{enumerate}
\item There is a binary $f \in \Pol(\Gamma)$ which is not dominated.
\item The relation $\Low$ is not pp-definable in $\Gamma$.
\item $e_<$ or $e_\leq$ is a polymorphism of $\Gamma$.
\item There is a binary $f \in \Pol(\Gamma)$ and endomorphisms $e_1,e_2 \in \End(\Gamma)$ such that 
\[ e_1(f(x,y)) = e_2(f(y,x)) \]
\item For all $c_1,\ldots,c_n \in \Gamma$ there is no clone homomorphism from $\Pol(\Gamma,c_1,\ldots,c_n)$ onto $\prcl$.
\item For all $c_1,\ldots,c_n \in \Gamma$ there is no continuous clone homomorphism from $\Pol(\Gamma,c_1,\ldots,c_n)$ onto $\prcl$. 
%\item There is no uniformly continuous h1 clone homomorphism from  $\Pol(\Gamma)$ onto $\prcl$
\item There is no pp-interpretation of $(\{0,1\}; \rm1IN3)$ in any expansion of $\Gamma$ by finitely many constants.
\end{enumerate}
\end{lemma}

\begin{proof} \ \\
The equivalences of the points (5)-(7) hold for all $\omega$-categorical structures $\Gamma$ and were discussed in Theorem \ref{theorem:doubleshrink}.\\
(1) $\leftrightarrow$ (2) This is the statement of Lemma \ref{lemma:lowdominated}.\\
(2) $\to$ (3): This is the statement of Theorem \ref{theorem:symcanexistence}.\\
(3) $\to$ (4): Set $f = e_<$ respectively $f = e_\leq$.\\
(4) $\to$ (5): If there are $e_1,e_2, f \in \Pol(\Gamma)$ satisfying the equation $e_1(f(x,y)) = e_2(f(y,x))$ then there are also such polymorphisms fixing finitely many elements $c_1,\ldots,c_n$. This is true for all $\omega$-categorical cores, see Lemma 82 of \cite{BPPPhyloCSP}. It follows that there is no clone homomorphism from $\Pol(\Gamma,c_1,\ldots,c_n)$ onto $\prcl$.\\
(7) $\to$ (2): This follows from the contraposition of Theorem \ref{theorem:lowhard}.
\end{proof}

With Lemma \ref{lemma:algdic} we are now able to show the following Theorem.

\begin{theorem} \label{theorem:algdic}
Let $\Gamma$ be a reduct of $\Po$ and let $\Delta$ be the model-complete core of $\Gamma$. Then the following are equivalent:
\begin{enumerate}
\item There is a binary $f \in \Pol(\Delta)$ and endomorphisms $e_1,e_2 \in \End(\Delta)$ such that 
\[e_1(f(x,y)) = e_2(f(y,x)) \]
or there is a ternary $f \in \Pol(\Delta)$ and endomorphisms $e_1,e_2, e_3 \in \End(\Delta)$ such that
\[e_1(f(x,x,y)) = e_2(f(x,y,x)) = e_3(f(y,x,x)).\]
\item There is a pseudo Siggers polymorphism, i.e. a function $f \in \Pol(\Delta)^{(6)}$ and endomorphism $e_1,e_2 \in \End(\Delta)$ such that
\[e_1(f(x,y,x,z,y,z)) = e_2(f(y,x,z,x,z,y)). \]
\item For all $c_1,\ldots,c_n \in \Delta$ there is no clone homomorphism from $\Pol(\Delta)$ onto $\prcl$.
\item For all $c_1,\ldots,c_n \in \Delta$ there is no continuous clone homomorphism from $\Pol(\Delta)$ onto $\prcl$.
%\item There is no uniformly continuous h1 clone homomorphism from $\Pol(\Gamma)$ onto $\prcl$
\item There is no pp-interpretation of $(\{0,1\},\rm 1IN3)$ in any expansion of $\Delta$ by finitely many constants.
\end{enumerate}
\end{theorem}

\begin{proof}
First of all we remark that the equivalence of the points (2)-(5) holds for all $\omega$-categorical core structures and was discussed in Theorem \ref{theorem:doubleshrink}.

In Theorem \ref{theorem:monoids} we saw that the model-complete core $\Delta$ is either equal to $\Gamma$ or a reduct of $(\Q,<)$ or $(\omega,=)$. 

Suppose the core $\Delta$ is a reduct of $(\Q,<)$ or $(\omega,=)$. We know from the analysis of temporal constraint satisfaction problems that then the statement is true: By Theorem 10.1.1. in \cite{manuelhabil} there is no pp-interpretation of $(\{0,1\},\rm 1IN3)$ in $\Delta$, if and only if an equation $e_1(f(x,x,y)) = e_2(f(x,y,x)) = e_3(f(y,x,x))$ holds in $\Pol(\Delta)$.

So let $\Gamma = \Delta$. By Lemma \ref{lemma:algdic} the equivalence (1)$\leftrightarrow$(4) holds when $<$ and $\bot$ are pp-definable in $\Gamma$. In the remaining cases $\End(\Gamma)$ is equal to $\overline{\langle \updownarrow \rangle}$, $\overline{\langle \circlearrowright \rangle}$ or $\overline{\langle \updownarrow, \circlearrowright \rangle}$ and we have a pp-interpretation of $(\{0,1\};\rm 1IN3)$ in an extension of $\Gamma$ with finitely many constants Theorems \ref{theorem:lowhard}, \ref{theorem:betwhard}, \ref{theorem:cyclhard} and \ref{theorem:septhard}.
\end{proof}

%We want to recall, that in all cases where $\Gamma$ was a model-complete core itself, we obtained definable subsets on which $<$ is pp-definable by fixing finitely many constants (cf. Section \ref{sect:cycl}). A natural question that arises from this is, whether this is true in a more general setting:
%
%{\color{blue} \begin{question}
%Let $\Delta = (D;R_1,\ldots,R_n)$ be a homogeneous structure and let $\Gamma$ be a reduct of $\Delta$ that is a model-complete core. Further let $\Aut(\Gamma) \not \supseteq \Aut(D;R_{i_1},\ldots, R_{i_n})$ for all $\{R_{i_1},\ldots, R_{i_n}\} \subsetneq \{R_1,\ldots,R_n\}$ , i.e. we need all the relations in $\Delta$ to define the relations in $\Gamma$. Do we always find constants $c_1,\ldots,c_k \in D$ and an orbit $X$ of $\Aut(\Gamma,c_1,\ldots,c_k)$ such that the restrictions $R_{1|X},\ldots,R_{n|X}$ are pp-definable in $(\Gamma,c_1,\ldots,c_k)$?
%\end{question}}

\subsection{A complexity dichotomy}

For the complexity of the CSPs of reducts of $\Po$ that are model-complete we have proven the following dichotomy:

\begin{theorem} \label{theorem:compdic}
Let $\Gamma$ be a reduct of $\Po$ in a finite relational language and a model-complete core. Under the assumption P$\neq$NP either
\begin{itemize}
\item one of the relations $\Low$, $\Betw$, $\Cycl$, $\Sept$ is pp-definable in $\Gamma$ and $\csp(\Gamma)$ is NP-complete or
\item $\csp(\Gamma)$ is tractable.
\end{itemize}
\end{theorem}

\begin{proof}
If $\Low$, $\Betw$, $\Cycl$ or $\Sept$ is pp-definable in $\Gamma$, the $\csp(\Gamma)$ is NP-complete by the Theorems \ref{theorem:betwhard}, \ref{theorem:septhard}, \ref{theorem:cyclhard} and \ref{theorem:lowhard}.

By Theorem \ref{theorem:monoids} the only remaining case is the one, where $<$ and $\bot$ are pp-definable, but $\Low$ is not. In this case $e_<$ or $e_\leq$ is a polymorphism of $\Gamma$ by Theorem \ref{theorem:symcanexistence}. Theorem \ref{theorem:horn2} then implies that the problem is tractable.
\end{proof}

\begin{corollary}
Let $\Gamma$ be a reduct of $\Po$ in a finite relational language. Under the assumption P$\neq$NP the problem $\csp(\Gamma)$ is either NP-complete or solvable in polynomial time. Further the ``meta-problem'' of deciding whether a given problem $\csp(\Gamma)$ is tractable or NP-complete, is decidable.
\end{corollary}

\begin{proof}
By Theorem \ref{theorem:monoids} we know that either $\Gamma$ is a model-complete core or $g_<$ or $g_\bot$ are endomorphisms of $\Gamma$. In the first case the dichotomy holds by Theorem \ref{theorem:compdic}, in the second case $\Gamma$ is homomorphically equivalent to a reduct of $(\Q,<)$ and the dichotomy holds by the result in \cite{Temp-SAT} respectively \cite{BKeqSAT}.

The main result in \cite{DecOfDefi} imply that it is decidable if the relations $<$, $\bot$, $\Low$, $\Betw$, $\Cycl$ or $\Sept$ are pp-definable in $\Gamma$. By Lemma \ref{lemma:relations} the question whether $\Gamma$ is model-complete core or not is then also decidable. By Theorem \ref{theorem:compdic} and Corollary 52 of  \cite{Temp-SAT} we have that the meta-problem is decidable.
\end{proof}

We finish with an algebraic version of our dichotomy that is a direct implication of Theorem \ref{theorem:algdic}:

\begin{corollary}
Let $\Gamma$ be a reduct of $\Po$ in a finite relational language and let $\Delta$ be its model complete core. Under the assumption P$\neq$NP either
\begin{itemize}
\item $\csp(\Gamma)$ is NP-complete and all finite structures are pp-interpretable in $\Delta$, extended by finitely many constant, or
\item $\csp(\Gamma)$ is tractable and the conditions (1)-(6) in Theorem \ref{theorem:algdic} hold. \hfill $\square$

\end{itemize} \end{corollary}

\bibliography{poset_sat}

\begin{thebibliography}{10}

\bibitem{anger1989lamport}
Frank~D Anger.
\newblock On {L}amport's interprocessor communication model.
\newblock {\em ACM Transactions on Programming Languages and Systems (TOPLAS)},
  11(3):404--417, 1989.

\bibitem{manuel-core}
Manuel Bodirsky.
\newblock The core of a countably categorical structure.
\newblock In {\em STACS 2005}, pages 110--120. Springer, 2005.

\bibitem{manuelhabil}
Manuel Bodirsky.
\newblock Complexity classification in infinite-domain constraint satisfaction,
  2012.
\newblock M\'emoire HDR \`a l'Universit\'e Paris 7, arXiv:1201.0856v7.

\bibitem{BCK-maximalCSP}
Manuel Bodirsky, Hubie Chen, Jan K{\'a}ra, and Timo von Oertzen.
\newblock Maximal infinite-valued constraint languages.
\newblock {\em Theoret. Comput. Sci.}, 410(18):1684--1693, 2009.
\newblock URL: \url{http://dx.doi.org/10.1016/j.tcs.2008.12.050}, \href
  {http://dx.doi.org/10.1016/j.tcs.2008.12.050}
  {\path{doi:10.1016/j.tcs.2008.12.050}}.

\bibitem{BKeqSAT}
Manuel Bodirsky and Jan K{\'a}ra.
\newblock The complexity of equality constraint languages.
\newblock {\em Theory of Computing Systems}, 43(2):136--158, 2008.

\bibitem{Temp-SAT}
Manuel Bodirsky and Jan K{\'a}ra.
\newblock The complexity of temporal constraint satisfaction problems.
\newblock {\em J. ACM}, 57(2):Art. 9, 41, 2010.
\newblock URL: \url{http://dx.doi.org/10.1145/1667053.1667058}, \href
  {http://dx.doi.org/10.1145/1667053.1667058}
  {\path{doi:10.1145/1667053.1667058}}.

\bibitem{BN-ppdefinability}
Manuel Bodirsky and Jaroslav Ne{\v{s}}et{\v{r}}il.
\newblock Constraint satisfaction with countable homogeneous templates.
\newblock {\em J. Logic Comput.}, 16(3):359--373, 2006.
\newblock URL: \url{http://dx.doi.org/10.1093/logcom/exi083}, \href
  {http://dx.doi.org/10.1093/logcom/exi083} {\path{doi:10.1093/logcom/exi083}}.

\bibitem{BP-reductsRamsey}
Manuel Bodirsky and Michael Pinsker.
\newblock Reducts of {R}amsey structures.
\newblock {\em AMS Contemporary Mathematics, vol. 558 (Model Theoretic Methods
  in Finite Combinatorics)}, pages 489--519, 2011.

\bibitem{BP-minimal}
Manuel Bodirsky and Michael Pinsker.
\newblock Minimal functions on the random graph.
\newblock {\em Israel Journal of Mathematics}, 200(1):251--296, 2014.

\bibitem{BP-graphsat}
Manuel Bodirsky and Michael Pinsker.
\newblock Schaefer's theorem for graphs.
\newblock {\em J. ACM}, 62(3):Art. 19, 52, 2015.
\newblock URL: \url{http://dx.doi.org/10.1145/2764899}, \href
  {http://dx.doi.org/10.1145/2764899} {\path{doi:10.1145/2764899}}.

\bibitem{topological-birkhoff}
Manuel Bodirsky and Michael Pinsker.
\newblock Topological {B}irkhoff.
\newblock {\em Transactions of the American Mathematical Society},
  367(4):2527--2549, 2015.

\bibitem{DecOfDefi}
Manuel Bodirsky, Michael Pinsker, and Todor Tsankov.
\newblock Decidability of definability.
\newblock {\em Journal of Symbolic Logic}, 78(4):1036--1054, 2013.

\bibitem{BJ-pointalgebras}
Mathias Broxvall and Peter Jonsson.
\newblock Point algebras for temporal reasoning: Algorithms and complexity.
\newblock {\em Artificial Intelligence}, 149(2):179--220, 2003.

\bibitem{Hodges}
Wilfrid Hodges.
\newblock {\em A shorter model theory}.
\newblock Cambridge University Press, Cambridge, 1997.

\bibitem{jeavons1998algebraic}
Peter Jeavons.
\newblock On the algebraic structure of combinatorial problems.
\newblock {\em Theoretical Computer Science}, 200(1):185--204, 1998.

\bibitem{lamport1986mutual}
Leslie Lamport.
\newblock The mutual exclusion problem: part {I}-a theory of interprocess
  communication.
\newblock {\em Journal of the ACM (JACM)}, 33(2):313--326, 1986.

\bibitem{BPPPhyloCSP}
Peter~Jonsson Manuel~Bodirsky and Trung~Van Pham.
\newblock The complexity of phylogeny constraint satisfaction.
\newblock {\em Submitted. Arxiv:1503.07310}, 2015.

\bibitem{Poset-Reducts}
P{\'e}ter~P{\'a}l Pach, Michael Pinsker, Gabriella Pluh{\'a}r, Andr{\'a}s
  Pongr{\'a}cz, and Csaba Szab{\'o}.
\newblock Reducts of the random partial order.
\newblock {\em Advances in Mathematics}, 267:94--120, 2014.

\bibitem{PPPS-a-new-transformation}
P\'{e}ter~P\'{a}l Pach, Michael Pinsker, Andr\'{a}s Pongr\'{a}cz, and Csaba
  Szab\'{o}.
\newblock A new operation on partially ordered sets.
\newblock {\em Journal of Combinatorial Theory, Series A}, 120:1450--1462,
  2013.

\bibitem{ramsey-poset}
M~Paoli, WT~Trotter~Jr, and JW~Walker.
\newblock Graphs and orders in {R}amsey theory and in dimension theory.
\newblock In {\em Graphs and Order}, pages 351--394. Springer, 1985.

\bibitem{pinsker-surveyCSP}
Michael Pinsker.
\newblock Algebraic and model theoretic methods in constraint satisfaction.
\newblock {\em arXiv preprint arXiv:1507.00931}, 2015.

\bibitem{BP-siggers}
Michael Pinsker and Libor Barto.
\newblock The algebraic dichotomy conjecture for infinite domain constraint
  satisfaction problems.
\newblock preprint on the authors website, January 2016.

\bibitem{schaefer}
Thomas~J Schaefer.
\newblock The complexity of satisfiability problems.
\newblock In {\em Proceedings of the tenth annual ACM symposium on Theory of
  computing}, pages 216--226. ACM, 1978.

\bibitem{szendreiclones}
{\'A}gnes Szendrei.
\newblock Clones in universal algebra, volume 99 of seminaires de mathematiques
  superieures.
\newblock {\em University of Montreal}, 61, 1986.

\end{thebibliography}

\end{document}